\documentclass[journal]{IEEEtran}
\IEEEoverridecommandlockouts
\usepackage{biblatex}
\usepackage{siunitx}
\usepackage{amsmath,amssymb,amsfonts,amsthm}
\usepackage{algorithm}
\usepackage[noend]{algpseudocode}
\usepackage{graphicx}
\usepackage{placeins}
\usepackage{makecell}
\newcommand{\joshua}[1]{{\color{black}{#1}}}
\newtheorem{theorem}{Theorem}

\usepackage{url}
\usepackage{cuted}
\newtheorem{assumption}{Assumption}
\usepackage{subcaption}
\usepackage{tikz}
\newcommand\copyrighttext{%
	\footnotesize  \copyright 2025 IEEE. Personal use of this material is permitted. Permission from IEEE must be obtained for all other uses, in any current or future media, including reprinting/republishing this material for advertising or promotional purposes, creating new
	collective works, for resale or redistribution to servers or lists, or reuse of any copyrighted
	component of this work in other works.
}
\newcommand\copyrightnotice{%
	\begin{tikzpicture}[remember picture,overlay]
		\node[anchor=south,yshift=10pt] at (current page.south) {\fbox{\parbox{\dimexpr\textwidth-\fboxsep-\fboxrule\relax}{\copyrighttext}}};
	\end{tikzpicture}%
}
\usepackage{textcomp}
\usepackage{bm}
\def\BibTeX{{\rm B\kern-.05em{\sc i\kern-.025em b}\kern-.08em
    T\kern-.1667em\lower.7ex\hbox{E}\kern-.125emX}}
\begin{document}

\title{Large-scale workflow placement in serverless
	computing using integer nonlinear programming
}

\author{Joshua Adamek$^{1}$, Natalie Carl$^{2}$, Trever Schirmer$^{2}$, Moritz Heinlein$^{1}$, David Bermbach$^{2}$
	and Sergio Lucia$^{1}$
\thanks{*This work has been funded by the Deutsche Forschungsgemeinschaft (DFG, German Research Foundation) - Projektnummer 495343202.}
\thanks{$^{1}$ Joshua Adamek, Moritz Heinlein and Sergio Lucia are with the Chair of Process Automation Systems at Technische Universität Dortmund, Emil-Figge-Str. 70, 44227 Dortmund, Germany.}%
\thanks{$^2$ Natalie Carl, Trever Schirmer, and David Bermbach are with the Scalable Software Systems Research Group at Technische Universität Berlin, Einsteinufer 17, 10587 Berlin.}
}

\maketitle
\copyrightnotice
\begin{abstract}
	
	Serverless edge computing has become a powerful cloud framework that enables the execution of large workflows without the need for the user to manage the underlying servers and edge devices. In this work, we address the challenge of deploying these workflows on \joshua{a large number of different }existing servers and edge devices such that monetary costs for the users and workflow evaluation times are minimized. To this end, the workflow and cloud node attributes are modeled in a mathematical framework. As a result, we \joshua{present a novel model of} 
	the optimal placement problem as a nonlinear integer program. To solve both the issues of scaling towards a larger number of cloud/edge nodes as well as decomposed knowledge of node attributes, we propose a \joshua{novel} decomposition strategy. In a case study, we show \joshua{the beneficial scaling properties of the decomposition approach and a mean improvement of $10\%$ against a simple deployment heuristic.}
\end{abstract}


\section{Introduction}
\label{sec:introduction}

Serverless computing has emerged as a promising cloud computing paradigm for Internet of Things (IoT) applications \cite{rajan_review_2020}. As the user interaction with the cloud platform is decoupled from the underlying physical servers, this offers flexibility and easier usage. This general idea motivates a Functions-as-a-Service (FaaS) architecture~\cite{yilmaz_multivocal_2020}, where deployment requires only functions and their data or software dependencies, while placement on physical devices is managed entirely by the cloud management system~\cite{eismann_state_2022}.  
FaaS platforms have gained attention in both research~\cite{barcelona-pons_faas_2019} and practice~\cite{sahraei_xfaas_2023}, particularly for IoT applications involving large numbers of dependent functions, also referred to as workflows~\cite{mahgoub2021sonic}. The decomposition of workflows into functions and their deployment represents one of the primary use cases for FaaS~\cite{wang2024briskchain}.
\begin{figure}[htbp]
	\centering
	\includegraphics[width=0.45\textwidth]{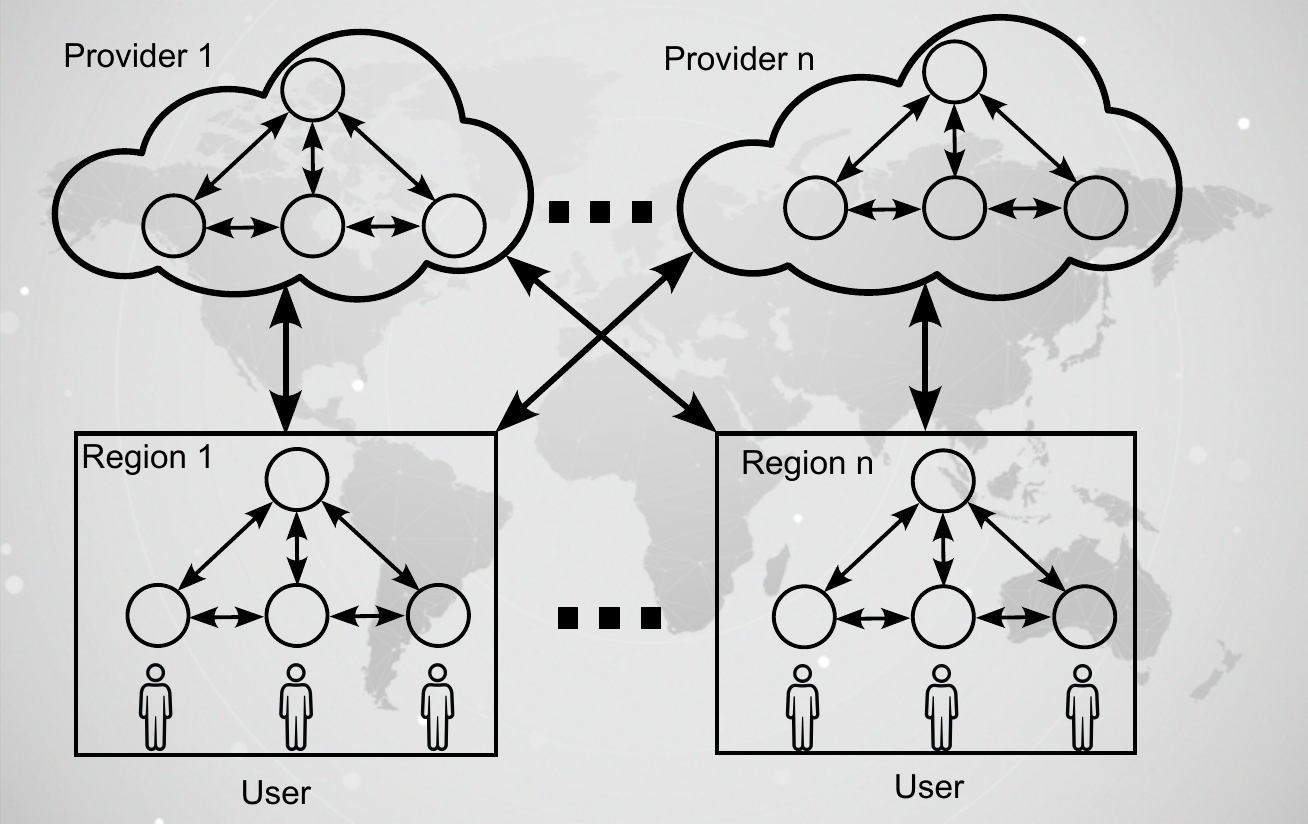}
	\caption{\joshua{Decentralized} serverless architecture considered in this work. \joshua{We consider a mixture of cloud, fog and edge computation nodes of different providers and locations as well as different attributes of the considered nodes.} \joshua{Cloud nodes are clustered by providers, edge/fog nodes by region of operation.}}
	\label{fig:overview}
\end{figure}
\joshua{
	Past literature distinguishes between the workflow deployment operation, commonly referred to as the \emph{placement problem} \cite{GHORBIAN2024103291}, and the scheduling operation, where a deployed workflow is selected for each user request. Both workflow placement and scheduling have been extensively studied in the literature \cite{GHORBIAN2024103291}. To efficiently utilize cloud resources, placement is typically optimized with respect to objectives such as latency, execution time, monetary cost, resource utilization, and energy efficiency \cite{GHORBIAN2024103291,ristov_2022,szalay_real-time_2023}.
	
	As described in \cite{GHORBIAN2024103291}, existing placement approaches can be broadly classified into machine learning-based \cite{xu2023stateful,WANG2024242,raza2023configuration}, heuristic-based \cite{DEHURY2024111179,smith}, and model-based methods \cite{costless,tarot,neptune,auctionwhisk,luo2024efficient,MOAKHAR2024103890,zahed2025efficient}. Heuristic approaches scale well to large workflow and cloud infrastructures but generally produce suboptimal solutions \cite{GHORBIAN2024103291}. Consequently, they have been widely adopted in practical serverless platforms such as XFaaS \cite{sahraei_xfaas_2023} and ECO \cite{alarbi2024eco}. Machine learning-based and model-based approaches, on the other hand, can provide optimal or near-optimal placements but face scalability challenges. Learning-based methods struggle with representing large workflow and infrastructure states and require substantial training data, typically targeting infrastructures with only $10$--$50$ nodes \cite{xie_workflow_2023}. Nevertheless, machine learning has also been successfully applied to large-scale serverless systems for improving placement and scheduling decisions based on runtime observations \cite{DEHURY2024111179}.
	
	Model-based approaches formulate workflow placement as an optimization problem based on mathematical models of workflows and cloud infrastructures. Existing formulations commonly represent workflows as directed acyclic graphs \cite{majewski} and cloud infrastructures as node-based networks \cite{xu2023stateful}. Two major challenges remain. First, the computational complexity of solving these optimization problems typically grows exponentially with the number of workflow functions and cloud nodes, limiting their applicability to large-scale deployments. While previous work has addressed the decomposition of large workflows \cite{yue}, cloud decomposition has received considerably less attention \cite{gallego2023machine}, particularly in the context of optimal placement. As FaaS infrastructures evolve from centralized cloud deployments toward geographically distributed cloud, fog, and edge environments spanning multiple providers \cite{alarbi2024eco,pfandzelter_tinyfaas_2020}, scalable optimization over large numbers of compute nodes becomes increasingly important.
	
	Second, the actual placement cost depends on the stochastic sequence of user requests and the resulting workflow selection during runtime \cite{palade}. Since modeling all possible request scenarios is computationally intractable \cite{palade,GHORBIAN2024103291}, many optimization-based approaches neglect the scheduling phase and instead optimize surrogate objectives such as latency or execution time \cite{auctionwhisk}. Furthermore, many optimization methods rely on shortest-path formulations for DAG-based workflows \cite{costless,wang}. Although recent work has proposed linear optimization models \cite{jia}, there remains a gap in modeling more complex workflow and infrastructure characteristics, including parallel workflow branches, data dependencies, and communication costs, within a nonlinear optimization framework. Advances in mixed-integer nonlinear optimization solvers, such as Gurobi \cite{gurobi} and BARON \cite{tawarmalani_convexification_2002}, now make such formulations computationally tractable while naturally supporting multi-objective optimization \cite{marler_weighted_2010}.
	
	This paper presents a mathematical framework for workflow placement in large-scale serverless architectures. To address the identified research gaps, the main contributions are:
	\begin{itemize}
		\item We propose a mathematical model for the workflow placement problem for which we can provide a hierarchical decomposition strategy that addresses the scalability and privacy challenges arising in large-scale decentralized serverless architectures. The proposed approach exhibits logarithmic growth in computational demand with respect to the number of cloud nodes.
		\item Into this a mathematical model we explicitly incorporate the workflow selection strategy. To account for waiting times during online execution, we introduce utilization-based costs for both compute nodes and workflow deployments.
		\item Additionally, we incorporate recently proposed improvements to the workflow modeling, such as data dependencies, parallel branches of workflows and data transmissions costs into this framework.
		\item The placement problem is formulated as a multi-objective optimization problem and we analyze the trade-offs between competing objectives using Pareto analysis.
		\item The scalability and performance of the proposed approach is evaluated through extensive simulation studies. Additionally, we analyze real world applicability of the approach through robustness analysis against incorrect parameter estimation during deployment.
	\end{itemize}
}

The remainder of this paper is organized as follows: Section II introduces the modeling framework, including attributes for workflows and nodes, Section III formulates the optimal placement problem and presents a decomposed solution enabling scalable algorithms for distributed clouds. Section IV provides case studies demonstrating advantages such as scalability to larger problems, adaptability in multi-objective settings, and comparison to heuristic approaches. Section V concludes with final remarks and outlook.

\section{Modeling of placement and selection of workflows}

\subsection{Workflow modeling}
In the following, we establish the modeling framework for the workflows considered in this work.
\joshua{Previous work commonly models workflows as directed acyclic graphs \cite{costless,wang}. We follow this general representation, but adapt it to the requirements of the proposed MINLP formulation for workflow placement in decentralized serverless architectures. In particular, the model is designed such that placement decisions, data transfers, dependencies between functions, and parallel workflow branches can be represented explicitly and later decomposed according to the regional structure of the infrastructure. We also include parallel workflow branches, as considered in \cite{wang}. The novelty of the modeling step therefore lies not in the graph-based workflow representation itself, but in its integration into a decomposable placement formulation for large-scale FaaS infrastructures.}
A workflow $s$ consists of $n_{\text{func}}$ functions. Each workflow is deployed to the cloud/edge structure $n_{\text{deployments}}$ times.  
Each function is stateless~\cite{eismann_state_2022}, meaning that all required data dependencies and inputs are provided externally, and a function is executed only once within a workflow. We assume that every function in a workflow is executed each time the workflow is invoked.
The workflow structure is illustrated in Fig.~\ref{fig:workflow}. We allow for $\joshua{n_{\text{branches}}}$ different branches within each workflow. To define a workflow, we introduce variables $b_{c,m}$, which equals one if function $m \in \{0,\ldots \, n_{\text{func}}-1\}$ belongs to branch $c \in \{0,\, \joshua{n_{\text{branches}}-1}\}$. Furthermore, we set $a_{m,\tilde m} = 1$ if function $\tilde m$ follows function $m$, indicating that function $m$ must be executed before data can be transferred between these two functions.
Each function has an estimated computational demand $R_{\text{func},m}$, defined by its computation time $t_{\text{func},m}$ and memory usage $\text{RAM}_m$. This demand is assumed to be fixed such that
\begin{align*}
	R_{\text{func},m}=\text{RAM}_m t_{\text{func},m}\, ,
\end{align*}
which is treated as a constant value.

\begin{figure}[htbp]
	\centering
	\includegraphics[width=0.35\textwidth]{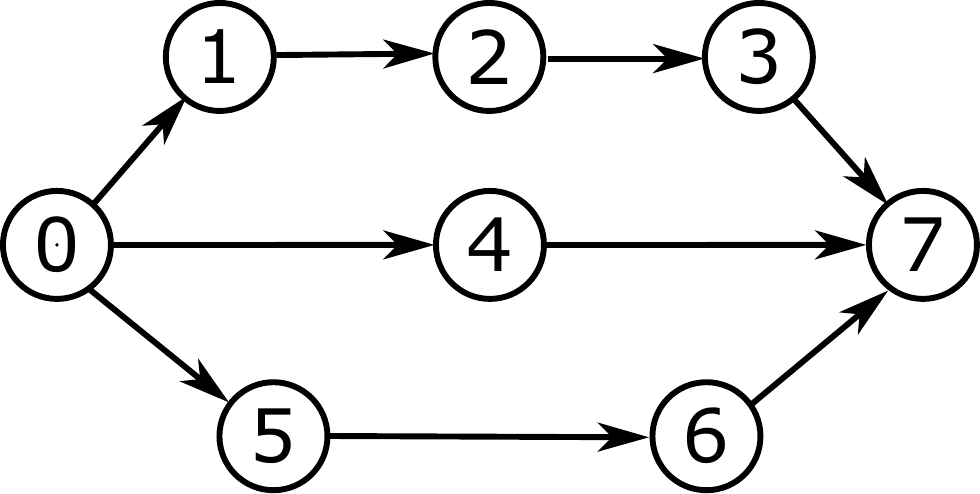}
	\caption{Structure of a workflow with eight functions and three branches. In this example, $b_{0,1}=1$ indicates that function~1 belongs to branch~0, while $a_{1,2}=1$ shows that function~2 succeeds function~1.}
	\label{fig:workflow}
\end{figure}
In addition to computational demand, each function sends data $d_{\text{sen\joshua{d}}}$ to its successor within the respective branch once execution finishes. The first function of a workflow receives input data $d_{\text{in}}$ from the user invoking the workflow. Figure~\ref{fig:workflow} illustrates these data flows required for completing the workflow. Furthermore, each function might require dependency data $d_{\text{\joshua{data}}}$, which are stored physically on nodes within the serverless architecture.

\subsection{Serverless infrastructure modeling}
\label{sub:serverless_modeling}
\joshua{We consider a hierarchical serverless edge-fog-cloud architecture rather than a fully peer-to-peer connected infrastructure. The modeled infrastructure contains all compute nodes that may host workflow functions, including edge nodes close to users, optional fog or regional nodes, and cloud nodes operated by one or multiple providers. User devices are modeled as request origins and are considered candidate nodes for function placement as we assume that these leave nodes explicitly provide computational resources. Edge and fog nodes may represent on-premises or regionally deployed resources, while cloud nodes are typically grouped according to their provider. The proposed model is agnostic to the ownership of individual nodes. Differences between providers, regions, and node types are captured through their latency, cost, speedup, and resource attributes. However, with this particular grouping, we ensure that every provider node attributes have to be shared to the overall architecture, providing a realistic real world setting for the serverless architecture.}
The task addressed by the placement problem is to assign each function $m$ of every workflow $n$ to a physical device within the serverless structure, referred to as a node. We define the serverless structure as a set of nodes $i \in \{0,\ldots,\joshua{n}_{\text{nodes}}\}$. These nodes are divided into edge nodes and cloud nodes. In this modeling framework, cloud nodes provide higher computational power, which can reduce the computation time of a function $m$. In our approach, we assume that executing functions on cloud nodes incurs monetary cost, whereas computations on edge nodes are considered free, but the modeling approach is not limited to this assumption. 
Furthermore, we consider multiple participants in the cloud environment. Each cloud node belongs to an entity (e.g. company or provider) that defines its attributes. For each provider, we assume a constant price to execute one megabyte for one second on its devices, denoted by $p_{\text{RAM}}$. This assumes a linear pricing model for the number of CPUs involved in executing a function. Additionally, we define a constant price for sending or receiving one megabyte of data as $p_{\text{send}}$, and a cost for transferring dependency data required for function execution as $p_{\text{\joshua{data}}}$ per megabyte.

\begin{figure}
	\centering
	\begin{subfigure}[t]{.21\textwidth}
		\centering
		\includegraphics[width=\textwidth]{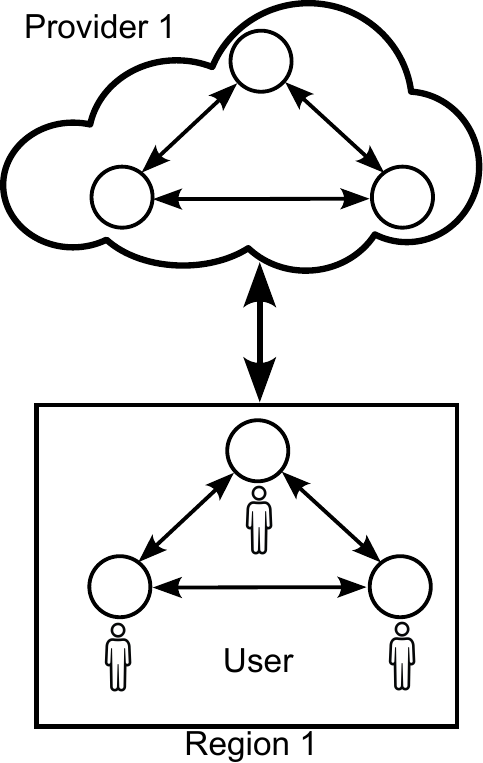}
		\caption{Assumed model of connections in the serverless architecture.}
		\label{fig:lat_corrrect}
	\end{subfigure}
	\begin{subfigure}[t]{.21\textwidth}
		\centering
		\includegraphics[width=\textwidth]{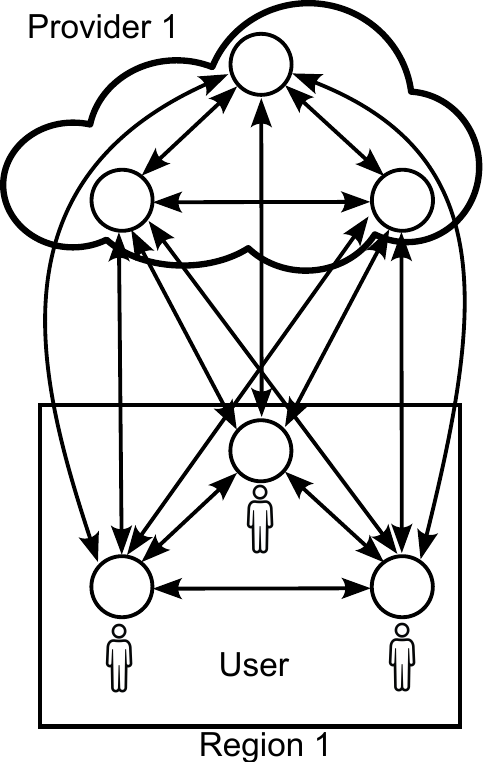}
		\caption{Non-scalable model of connections in the serverless architecture.}
		\label{fig:lat_incorrect}
	\end{subfigure}
	\caption{General connection structure of the serverless architecture. Cloud providers and edge nodes are distinguished. Requests from users originate from edge nodes only. Cloud nodes are clustered by provider or region. Connections between providers or regions exist through single links, as shown on the left. A fully connected topology, as shown on the right, is neither efficient nor realistic when different providers are involved.}
	\label{fig:lat}
\end{figure}
\joshua{The resulting topology is represented as a hierarchy of regions and providers. Nodes within the same region or provider can communicate through local links, whereas communication between different regions or providers is routed through representative head nodes. These head nodes may correspond to physical gateways, provider endpoints, or logical aggregation points used by the orchestration layer. Consequently, the architecture is not assumed to be fully connected at the physical node level. Instead, the hierarchical representation reflects practical decentralized deployments in which detailed information is locally available within regions or providers, while only aggregated information is exchanged between them.}
Figure~\ref{fig:lat} illustrates the proposed structure and connectivity of the serverless architecture using a simplified example. The bottom level represents user nodes at the edge from which requests originate during workflow execution.
Each connection between node $i$ and node $j$ is characterized by an estimated latency $L_{i,j}$. As shown in Fig.~\ref{fig:lat_incorrect}, an idealized setting would allow all nodes to be fully connected so that data could be transferred freely among them. In practice, especially with many providers and users, a separation occurs as depicted in Fig.~\ref{fig:lat_corrrect}. It is reasonable to assume that regions or providers can communicate only through single head node links connecting their respective top-level representatives.

This leads naturally to a hierarchical decomposition into regions and providers. In the simple example, there are two levels in the hierarchy: level zero contains two head nodes (region one and provider one), while level one includes three subordinate nodes per region/provider. We now establish notation for indexing these hierarchical levels in general form with levels indexed by $l \in \{0,\ldots,n_l-1\}$. At each level $l$, we consider $n_{\text{regions},l}$ regions/providers indexed by $r \in \{0,\ldots,n_{\text{regions},l}-1\}$, each containing $n_{\text{nodes},l,r}$ individual nodes indexed by $i \in \{0,\ldots,n_{\text{nodes},l,r}-1\}$. The number of regions at level $l$ is therefore the number of nodes at level $(l-1)$:
\begin{align*}
	n_{\text{regions},l}=\sum_{s=0}^{n_{\text{regions},l-1}}n_{\text{nodes},l-1,r}\, .
\end{align*}

In this setting, latencies are defined for every connection between two nodes within the same region such that $L_{l,r,i,j}$ denotes the latency between nodes $i$ and $j$ in region $s$ at level~$r$.
To navigate within the regional structure we define the following four functions: Given a node $i$ in level $l$ (that is the overall node number in the entire level over all regions in this level), we define the function $r=\text{curr\_region}(l,i)$ that outputs the current region of this node, while $\tilde i=\text{curr\_node}(l,i)$ is the index of the node within this region (that is the local node number in this region starting from the index zero in each region). Likewise, for a region $s$ in level $r$, we define the successor node $\tilde i=\text{succ\_node}(l,r)$ in level $l-1$, which is the index of the node in the successor region $\tilde r=\text{succ\_region}(l,r)$.
The latencies between two nodes $i$ and $j$ can be calculated according to Algorithm~\ref{alg:lat}.
\begin{algorithm}[tbp]
	\caption{Latency calculation for regionally structured serverless architectures.}\label{alg:lat}
	\begin{algorithmic}
		\Require Indices $i$ and $j$ of the physical devices
		\State $L_{i,j}=0$ and $l=n_l-1$
		\State $\tilde i=\text{curr\_node}(l,i)$, and $\tilde j=\text{curr\_node}(l,j)$
		\State $s_i=\text{curr\_region}(l,i)$ and $s_j=\text{curr\_region}(l,j)$ 
		\While{$s_i\neq s_j$} 
		\State $L_{i,j}=L_{i,j}+L_{l,r_i,0,\tilde i}+L_{l,r_j.0,\tilde j}$
		\State $\tilde i=\text{succ\_node}(l,r_i)$ and $\tilde j=\text{succ\_node}(l,r_j)$
		\State $r_i=\text{succ\_region}(l,r_i)$ and $r_j=\text{succ\_region}(l,r_j)$
		\EndWhile
		\State $L_{i,j}=L_{i,j}+L_{l,r_i,\tilde i,\tilde j}$
	\end{algorithmic}
	
\end{algorithm}
\joshua{\subsection{Operation modeling}
	We model the considered architecture as a serverless FaaS environment, where resource management, orchestration, deployment, and workflow selection are handled by the serverless platform. We assume that multiple cloud providers and edge or fog regions are jointly organized through a logical orchestration layer. During placement, a deploying user requests the deployment of one or multiple workflows. The orchestration layer then decides on which nodes the individual functions of each workflow are placed. In a centralized setting, this decision is made by one logical entity with access to the complete infrastructure model. In the decomposed setting, placement decisions are distributed across the hierarchy, such that regional or provider head nodes solve local subproblems using local information.
	For every individual deployed workflow that is for the number of instances of one specific workflow we solve one placement problem, to reduce the complexity of the underlying optimization problem.
	
	During runtime, users invoke workflows through requests originating from leaf nodes of the architecture. These leaf nodes may also provide computational resources and can therefore be part of the placement decision. For each request, the orchestration layer applies a selection policy that determines which deployed workflow instance is used. Thus, the user only submits deployment and invocation requests, while placement, selection, and execution are managed by the serverless platform.
	After the initial placement we assume no rescheduling of the placed workflow instances during the online scheduling phasse.
}

\section{Solution strategy}

\subsection{Decision variables and assumed problem structure}
\joshua{In the following we setup the optimization problem for the placement of the workflows in the proposed serverless architecture. There already exists linear mixed-integer optimization model formulations \cite{jia}. However in this work, we consider more complex workflow and infrastructure characteristics, including parallel workflow branches, data dependencies, and communication costs, which require formulating the optimization as a mixed-integer nonlinear problem. Since mixed-integer nonlinear programs are computationally demanding, we establish a decomposition approach to enable deployment on larger serverless architectures. }To set up the placement problem, the optimization variables defining the placement must be established.
Since each workflow $t$ is deployed separately, we omit the index $t$ in our problem formulation for brevity.
Each function $m$ of workflow $n$ must be placed on a node $i$ within the serverless architecture. Therefore, we define the decision variable $P_{n,m,i}$ to be one if this function is placed on node~$i$, and zero otherwise. For every function, exactly one node is selected, such that
\begin{align*}
	\sum_{i}P_{n,m,i}=1\, ,
\end{align*}
for all functions $m$ in workflow~$n$.
The main difficulty in describing the cost function arises from the fact that total cost does not directly result from workflow placement itself but rather from user selection of workflows during execution. Consequently, we aim to minimize the expected selection cost incurred by users invoking workflows. 
To model user behavior, we define a variable $H_{k,n}$ indicating which workflow is selected by user~$k$. In general, this variable is not known in advance. It represents a runtime decision made during operation. Optimal choices may differ depending on request conditions. Modeling such behavior leads to a two-stage stochastic optimization problem~\cite{schultz_twostage_1996}
, where placement variables $P_{n,m,i}$ would need to be optimized for all possible realizations of $H_{k,n}$. Since integer-based two-stage stochastic problems are computationally demanding~\cite{schultz_twostage_1996}
, they are unsuitable for large-scale serverless architectures.

We therefore assume a deterministic selection policy $H_{k,n}$, which we determine prior to execution and include as decision variables in the placement problem, ensuring that exactly one workflow is selected for each user request:
\begin{align*}
	\sum_{n}H_{k,n}=1\, ,
\end{align*}
for every possible user~$k$. In turn, this policy cannot react to the stochastic nature of the user requests, leading to possible blocking of resources or deployed functions, amplifying the need to address this issue in the placement problem as proposed in the next subsection.

\subsection{Problem definition}
\label{sub:problem_statement}
After modeling the workflows and nodes in the serverless architecture, we next define costs and constraints for the placement problem. 
We place each workflow separately, that is, each workflow $\joshua{s}$ is placed by solving the following optimization problem.
We consider two competing objectives that influence the placement: the time required to evaluate a workflow and the monetary cost incurred by a user when executing it.
\begin{table}[htbp]
	\centering
	\caption{Overview of the model parameters used in the main optimization problem~\eqref{eq:main}.}
	\label{tab:parameters}
	\begin{tabular}{ll}
		\hline
		Notation&Definition\\
		\hline
		$k$& Index used to address leaves\\
		$i$,$j$& Indices used to address general nodes in the cloud\\
		$l$& Index used to address levels in the serverless architecture\\
		$m$& Index used to address functions within a workflow\\
		$n$ &Index used to address different deployments\\& of the same workflow\\
		$r$ &Index used to address a region in a level\\& in the serverless architecture\\
		$c$& Index used to address the branches of the workflow\\
		$H_{k,n}$& Variable selecting workflow $n$ for a request from leaf $k$\\ 
		$P_{n,m,i}$&Variable placing function $m$ \\ & from workflow $n$ on the node $i$\\ 
		$C_{m,i}$&Cost to execute function $m$ on node $i$\\
		$D_{m,i,j}$& Cost to send data generated by function $m$\\& from node $i$ to node $j$\\
		$T_{n}$&Maximum time to execute workflow $n$\\
		$T_{m,i}$ &Time to execute function $m$ on node $i$\\
		$L_{\joshua{l,r},i,j}$& Latency between node $i$ and $j$\\& in region $s$ in level $r$ of the serverless architecture.\\ & If only one region is considered this is shorten to $L_{i,j}$.\\
		$p_{\text{RAM}}$& Cost to execute one MB of RAM for one second\\
		$p_{\text{send}}$&Cost to send $1$ MB of output data from a function\\
		$p_{\text{data}}$&Cost to transfer $1$ MB of data\\& needed for the execution of a function\\
		$d_{\text{data}}$&Amount of data\\& needed for the execution of a function\\
		$d_{\text{send}}$&Amount of output data send from a function\\
		$t_{\text{func},m}$&Estimated running time of function\\& $m$ on a normal device\\
		$\text{RAM}_m$& Estimated RAM usage of function $m$\\
		$R_{\text{func},m}$&Constant computational demand of function $m$\\
		$n_{\text{func}}$& Number of functions in a workflow\\
		$\joshua{n_{\text{branches}}}$& Number of branches of workflow\\
		$n_{\text{deployments}}$&Number of deployments of workflow $s$\\
		$b_{c,m}$&Indicates if function $m$ is in branch $c$\\
		$a_{m,\tilde m}$&Indicates if function $\tilde m$ is a successor of function $m$\\
		$x_{i,m}$&Estimated speedup of large functions \\& deployed on cloud node $m$\\
		$n_{l}$&Number of levels considered in\\& the serverless architecture\\
		$n_{\text{regions},l}$&Number of regions considered in level $l$\\& in the serverless architecture\\
		$n_{\text{nodes},l,r}$& Number of nodes in region $r$ on level $l$\\& in the serverless architecture\\
		\hline
	\end{tabular}
\end{table}

Using the modeling introduced in the previous section, these costs can be expressed as functions of the placement and selection decision variables, as well as parameters describing nodes and workflows (summarized in Table~\ref{tab:parameters}). We denote by $T_{n,c}$ the time required to evaluate branch~$c$ of workflow~$n$. The overall evaluation time depends on the slowest branch among all $\joshua{n}_{\text{branches}}$ branches of a workflow. To include this dependency in the optimization problem, we introduce an auxiliary variable $T_n$, constrained by
\begin{align*}
	T_n \geq T_{n,c} \quad \forall c \in \{0,\joshua{n}_{\text{branches}}\}\, .
\end{align*}
Each branch time $T_{n,c}$ is computed as the sum of function evaluation times for all functions~$m$ within that branch and latencies between nodes hosting consecutive functions. Consequently, we can formulate
\begin{align*}
	T_{n,c} &= \sum_{m}\sum_{i}(b_{c,m}P_{n,m,i}t_{m,i}\\
	&+ \sum_{\tilde m}\sum_{j}P_{n,m,i}P_{n,\tilde m,j}a_{m,\tilde m}b_{c,m}L_{i,j})\, .
\end{align*}
\joshua{The execution time $t_{m,i}$ of function $m$ on node $i$ is determined by the base execution time $t_{\text{func},m}$ multiplied with a potential speedup $x_{m,i}$ of executing the function on node $m$
	\begin{align*}
		t_{m,i}=x_{m,i}t_{\text{func},m}\, .
	\end{align*}
}
Note that latency introduces bilinear terms into this expression, as the placement variables $P_{n,m,i}$ and $P_{n,\tilde m,j}$ are multiplied.
\begin{figure}[htbp]
	\centering
	\includegraphics[width=0.42\textwidth]{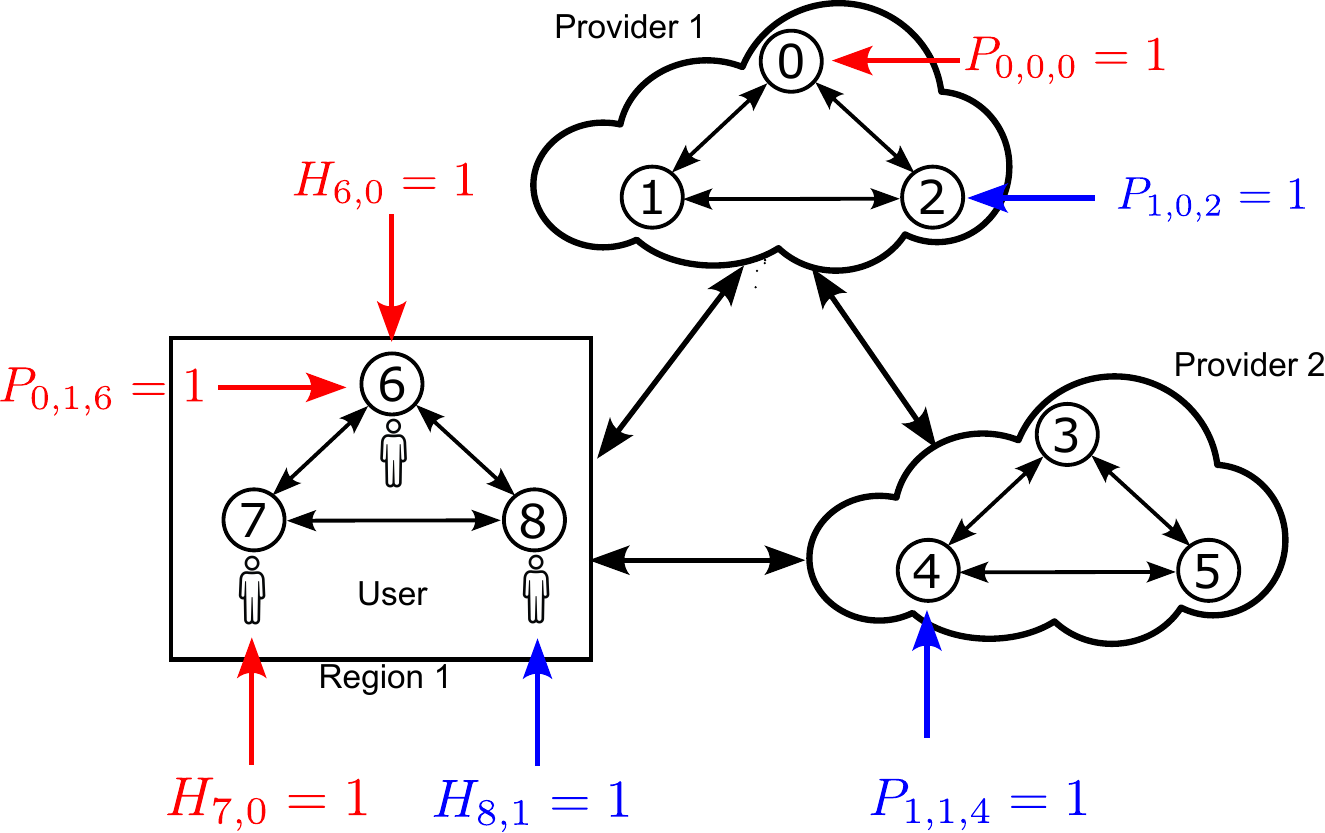}
	\caption{Overview of the centralized optimization strategy on the example of a workflow consisting of two functions placed two times on an architecture with three edge nodes and two providers with three physical devices. The placement variables $P$ show which function of which workflow is placed on which device, while the selection variable $H$ defines which workflow will be selected from a user node. The first two users (node 6 and 7) select the first workflow (variables in red) while the third user (node 8) select the second workflow (variables in blue).}
	\label{fig:example_centralized}
	
\end{figure}
Analogously, monetary costs consist of two components: a linear part representing direct execution costs $C_{m,i}$ for running function~$m$ on node~$i$, and a bilinear part describing data transfer costs between adjacent functions deployed on different nodes. We denote $C_{m,i}$ the sum of the execution cost and the potential additional cost for transferring data to node~$i$
\begin{align*}
	C_{m,i}= p_{\text{RAM},i}R_{\text{func},m}+p_{\text{data},i}d_{\text{data},m}\, .
\end{align*}
If the transfer data is already available on node $i$, we formulate $p_{\text{data},i}=0$.
Additionally, data is transferred between two functions located on different nodes, represented by
\begin{align*}
	D_{m,i,j}= (p_{\text{send},i}+p_{\text{send},j})d_{\text{sen\joshua{d}},m}\, .
\end{align*}
\joshua{Therefore, we can introduce the bilinear monetary variable for sending data as
	\begin{align*}
		S_{n,m,i}=\sum_{\tilde m}\sum_{j} P_{n,m,i}P_{n,\tilde m,j}a_{m,\tilde m} D_{m,i,j} \,\, \forall (n,m,i).
	\end{align*}
}

\begin{figure*}
	\begin{align}
		\label{eq:main}
		\begin{split}
			\underset{P,T,H,S}{\operatorname{min}} \, &\underbrace{w_1\sum_{k}\sum_{n}\lambda_k H_{k,n}\left( P_{n,m_{\text{start}},i}D_{\text{in},i}+P_{n,m_{\text{end}},i}D_{m,i,k}+\sum_{i}\sum_{m}P_{n,m,i}C_{m,i}+S_{n,m,i}\right)}_{\text{Monetary costs }M}\\
			+&\underbrace{w_2\sum_{k}\sum_{n}\lambda_k H_{k,n}\left( T_{n}+\sum_{i}\left(P_{n,m_{\text{start}},i} L_{k,i}+P_{n,m_{\text{end}},i}L_{k,i}\right)\right)}_{\text{Basic Time cost }B}
			+\underbrace{w_3\left(\sum_{n}\sum_{m}(u_{n,m})^2+\sum_{i}M_i^2\right)}_{\text{Utilization cost }U}\\
			&\text{subject to:}\\
			&\underbrace{\sum_{i} P_{n,m,i}=1 \,\,\forall (n,m)\quad \sum_{n} H_{k,n}=1 \,\,\forall (k)}_{\text{Placement and selection constraint}} \quad\underbrace{S_{n,m,i}=\sum_{\tilde m}\sum_{j} P_{n,m,i}P_{n,\tilde m,j}a_{m,\tilde m} D_{m,i,j} \,\, \forall (n,m,i)}_{\text{Monetary bilinear variable}}\\
			&\underbrace{T_n\geq\sum_{m}\sum_{i} (b_{c,m}P_{n,m,i}t_{m,i}+\sum_{\tilde m}\sum_{j} P_{n,m,i}P_{n,\tilde m,j}a_{m,\tilde m} b_{c,m} L_{i,j})\,\, \forall (n,c)}_{\text{Maximum branch time constraint}}\\
			&\underbrace{M_i=\frac{\sum_k\sum_n\sum_m \lambda_k H_{k,n} P_{n,m,i}R_{\joshua{\text{func},m}}}{\text{RAM}_{\text{max},i}}+M_{0,i}\, \forall (i)}_{\text{Node utilization}}\quad \underbrace{u_{n,m}=\sum_k \lambda_k H_{k,n} P_{n,m,i}t_{m,i}\, \forall (\joshua{n,m})}_{\text{Function utilization}}\\
		\end{split}
	\end{align}
\end{figure*}
In addition to these two primary cost components, time and monetary, we must also account for communication costs between users and nodes executing either the initial or final function of each workflow. These contributions are added respectively to both time-related and monetary terms.
Furthermore, all these costs results from invoking these workflows from user requests. We assume that the user randomly requests workflows with a rate of $\lambda_k$ for every leaf node $k$. To get the correct expected cost, we sum over user requests by weighting each monetary cost and evaluation time by the request rate $\lambda_k$ for every user.

These two cost factors are exact if every resource and function is always free to use, that is, we do not expect waiting time in the scheduling approach due to resource constraints.
As this is not realistic, we introduce an additional cost term addressing resource limitations. Because exact sequences of user requests cannot be anticipated during runtime, imposing strict constraints on workflow or node usage would lead to conservative solutions. Instead, we define a third penalty term that discourages excessive utilization of specific workflows or nodes within the serverless architecture.
Additionally, modeling the exact waiting time anticipated due to the utilization of nodes and functions is equally unrealistic.
Instead, we penalize the unequal placement of functions upon the serverless architecture, especially edge nodes that are resource constraint by a maximum RAM $\text{RAM}_{\text{max},i}$\joshua{, that is, we assume function evaluations on these edge nodes are blocked if the RAM reaches the maximum RAM constraint.} 
We penalize the expected utilization of each node $M_i$ where
\begin{align*}
	M_i=\frac{\sum_k\sum_n\sum_m \lambda_k H_{k,n} P_{n,m,i}R_{\joshua{func,}m}}{\text{RAM}_{\text{max},i}}+M_{0,i}\, .
\end{align*}
The utilization of this resource is dependent on the usage of this node over all $s$ workflows. To mitigate this, we track the utilization by adding the current utilization to $M_{i,0}$, which in turn is the summed utilization of previous placed workflows.
Furthermore, blocking of a function happens if a specific function is already used through a different request. We therefore penalize the function utilization given as
\begin{align*}
	\joshua{u}_{n,m}=\sum_{k}\sum_{i}\lambda_kH_{k,n}P_{n,m,i}t_{m,i}\, .
\end{align*}
\joshua{In this work we only assume request blocking because of exceeding the RAM limit and when the specific deployment of a function is currently utilized by another user. Therefore, reducing the utilization of a node automatically reduces the chance of blocking and in return will reduce the waiting time. We will test this hypothesis in the case study and additionally evaluate how this additional cost influence the optimization problem.}

As illustrated in Figure~\ref{fig:example_centralized}, the total cost is evaluated as the summation over all potential user requests.
The three individual cost components, monetary, temporal, and utilization-related, together form a multi-objective optimization problem. We solve this problem by employing a weighted-sum approach~\cite{marler_weighted_2010}, assigning each objective a corresponding weight.
\joshua{The weights dimensions are chosen such that the cost function is dimensionless ($w_1 [\$^{-1}],\, w_2 [s^{-1}],\, w_3 [-]$).}
Consequently, the complete optimization problem can be formulated as in~\eqref{eq:main}.

Every cost is nonlinear, due to the multiplication of the selection and placement optimization variables. Additionally, nonlinear terms are added in the constraints as we model time and monetary costs based on the interaction of adjacent functions on possible different nodes. We further introduce continuous variables such as the monetary bilinear term $S$ and the basic time variable $T$. The problem we are solving is therefore classified as a mixed-integer nonlinear  program. This problem formulation scales exponentially in computation time for an increasing number of nodes and workflow functions, which is especially difficult for a serverless architecture with many nodes. Therefore, and considering the already established decompositions that are build into the modeling of the serverless architecture, we propose a decomposition approach that scales well for a larger number of nodes in the next subsection.
\joshua{We do not directly model time-out of resources or deadlines, but the heuristic nature of the proposed utilization cost extends to these problems.
	Future work will address these as additional constraints in the optimization problem formulation.}

\subsection{Decomposition strategy}

There are two main reasons why solving~\eqref{eq:main} is computationally intractable in practice. First, it requires complete knowledge of all costs, speedups, and latencies across the entire serverless architecture on a single node. This assumption is unrealistic since nodes may belong to different providers that do not share all information, and collecting such statistics would accumulate large amounts of data on one node. Second, computation time scales exponentially with the number of nodes considered. As described in Subsection~\ref{sub:serverless_modeling}, direct communication among all nodes as a fully connected network is improbable. Instead, it is more realistic to divide the architecture into smaller regions represented on the higher level as a single node.
We extend this idea to the optimization problem itself. On each hierarchical level, placement decisions for functions within a region are performed by the respective head node that possesses local information such as latencies, costs, and speedups for its region. 
These regions are chosen such that node characteristics (speedup factors, type of node, region) remain homogeneous within each subset, and edge nodes are grouped only with other edge nodes.
This approach addresses both challenges: only local information is required, and regions remain small enough for efficient computation. An example of the structure of the decomposition is illustrated in Figure~\ref{fig:example_decomposed}. 
The goal of the decomposition approach is to solve a variant of~\eqref{eq:main} independently within each region at every level while maintaining proximity to the solution obtained when considering all nodes simultaneously.
Therefore, we now reformulate the decision variables to $H_{l,r,k,n}$ and \joshua{$P_{l,r,n,m,i}$} as we consider $r \in \{0,\ldots,n_{\text{regions},l}-1\}$ optimization problems in level $l \in \{0,\ldots,n_{l}-1\}$.
\begin{figure}[htbp]
	\centering
	\includegraphics[width=0.42\textwidth]{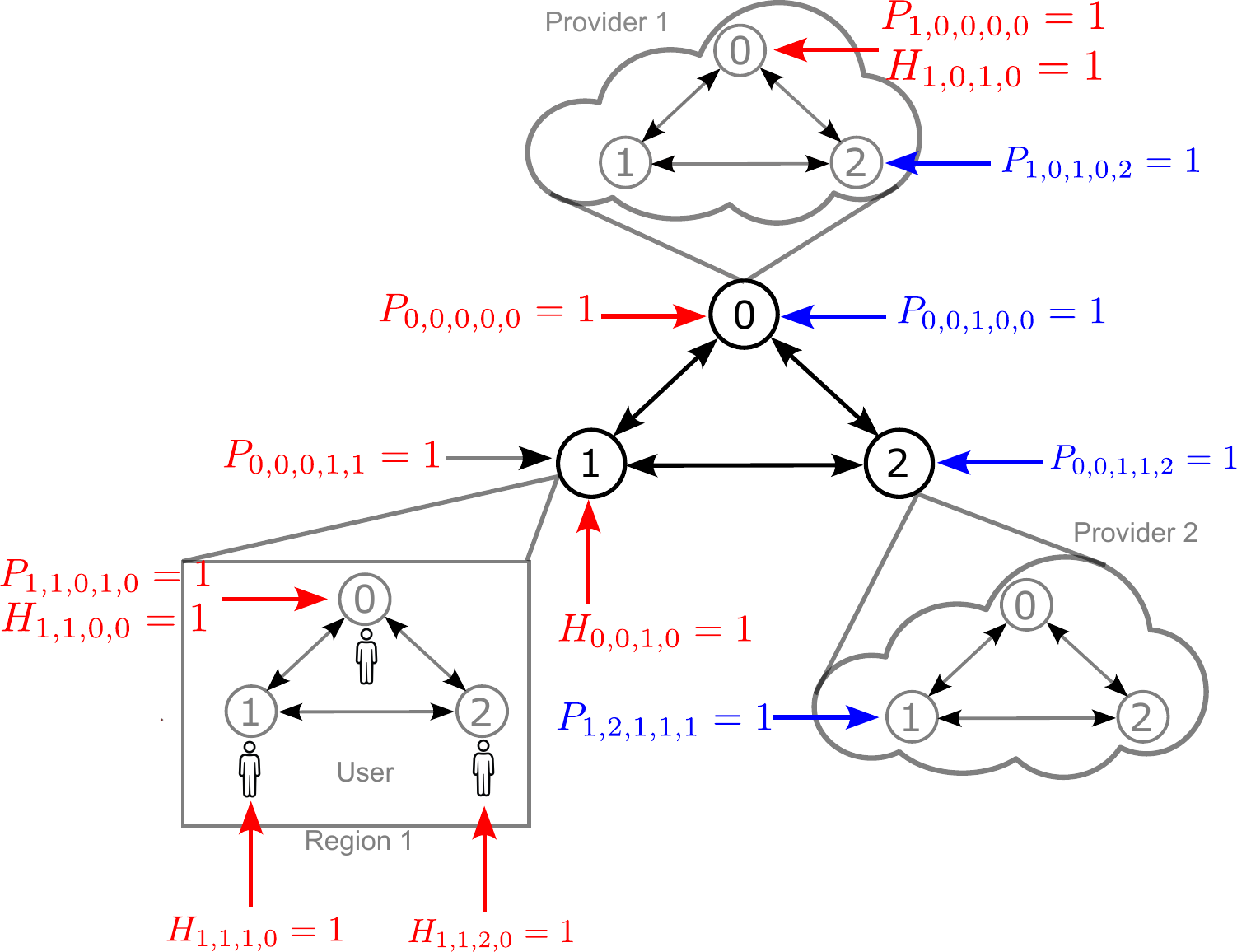}
	\caption{Decomposed approximation based on regions and providers shown for three nodes per provider and three edge-region nodes. The one workflow in this example is placed twice $n_{\text{deployments}=2}$. Each region or provider is represented by a single node in the main problem, where placement variables $P$ and selection variables $H$ are determined. Functions placed on representing nodes are then resolved within subproblems for their respective subsets. In contrast to the centralized approach, selection can be performed only once per region. Thus, only one workflow can be selected for all three edge nodes collectively in this example.}
	\label{fig:example_decomposed}
\end{figure}
Figure~\ref{fig:example_decomposed} illustrates the decomposition strategy. We start at level $n_l-1$, where physical nodes exist. For the next level, groups of edge-node regions and cloud providers are aggregated into single representative nodes. This results in several regions, each having a head node responsible for solving its local subproblem.
Analogous to the main problem formulation, user requests must be included for every subproblem. Any node consisting solely of edge devices or representing physical edge hardware remains treated as a leaf node from which requests originate.
Therefore, the user request rate is calculated by summing over all rates in the lower level
\begin{align}
	\label{eq:add_cons_1}
	\lambda_{l-1,\text{succ\_node}(\joshua{l,}k)}=\sum_{k}\lambda_{l,k}\, .
\end{align}

Each region of nodes at level~$l$ corresponds to a single representative node at level~$(l-1)$. Therefore, placements within regions can be tracked by determining whether function~$m$ of workflow~$n$ was placed on its predecessor node at level~$(l-1)$. Naturally, we start solving the placement problem at level~0 and proceed through successive levels by providing the upper level placements to the lower level head nodes. Although this introduces an exponential increase in subproblem count across levels, each subproblem operates independently on its own head node and can thus be computed in parallel.
Two additional constraints must be introduced to correctly perform decomposition:
First, if certain functions or entire workflows have already been placed elsewhere, that is, if these functions were not assigned to their predecessor node at level $(l-1)$, they cannot be reconsidered within the current subproblems. Because communication routes always pass through head nodes, placing such functions virtually on their head node has an equivalent effect as placing them outside their subset entirely (see Figure~\ref{fig:example_decomposed}). \joshua{These virtual placements do not represent functions that are physically executed on the local head node, but only serve as placeholders for workflow parts that are outside the current region and interact with it through the head-node interface.} Therefore, we add a function $m$ into the set $\mathcal{M}\joshua{_n}$ if $P_{l-1,\text{succ\_region}(\joshua{l,}r),n,m,\text{succ\_node}(\joshua{l,}r)}$=0 and add the following constraints
\begin{align}
	\label{eq:add_cons_3}
	P_{l,r,n,m,0}=1\, \, \forall m\in \mathcal{M}\joshua{_n}\, .
\end{align}
As these are virtual placements, they cannot be included in the calculation of the selection cost, especially the function placement penalization $M_{l,r,0}$ on the head node.
To account for this number of virtual placements in the selection cost, we have to adapt the function term $M_{l,r,0}$ for all subproblems
\begin{equation}
	\label{eq:add_cons_4}
	\begin{aligned}
		M_{l,r,0}&=\frac{\sum_k \sum_n \sum_{m^\star} \lambda_{l,r,k} H_{l,r,k,n}P_{l,r,n,m^\star,0}R_{\joshua{func,}m^\star}}{\text{RAM}_{\text{max},0}}\\
			&+M_{l,r,0,0}\, ,
	\end{aligned}
\end{equation}

where $m^\star\not\in \mathcal{M}\joshua{_n}$.
If complete workflows have been deployed elsewhere, we reduce $n_{\text{deployments}}$ by one accordingly. This pruning strategy reduces computational effort since regions without placed functions do not need to perform further calculations.
\joshua{Under the assumed hierarchical communication model, this pruning step does not introduce additional local cost by itself.}
The situation differs for selection decision variables $H_{l,r,k,n}$. Because upper-level representations also aggregate multiple edge regions into regional entities, selection decisions become coupled among users belonging to those regions (see Figure~\ref{fig:example_decomposed}). To obtain correct cost estimation during selection evaluation, we must assume identical workflow choices across all user requests originating from one region. Hence, recomputation of $H_{l,r,k,n}$ beyond level $l>0$ is omitted. 
Again, we have to distinguish between leaf nodes that stem from real users, where if in the upper level $\text{succ\_node}(\joshua{l,}r)$ is a leaf node and $H_{l-1,\text{succ\_region}(\joshua{l,}r),\text{succ\_node}(\joshua{l,}r),n}=1$ we set
\begin{align}
	\label{eq:add_cons_5}
	H_{l,r,k,n}=1\, ,
\end{align}
for all leaf nodes $k$ participating in the next-level subproblems.
For the head node zero, we have to additionally account for workflows that are possibly invoked through nodes from other regions. As these information exchanges run through this main node, we treat the main node as the leaf for these additional nodes. Therefore, we sum the selections $H_{l-1,\text{succ\_region}(\joshua{l,}r),k,n}=1$ for all leaves that are not the successor of this optimization problem that is
\begin{align}
	\label{eq:add_cons_6}
	H_{l,r,0,n}&=\Gamma\left(\sum_k H_{l-1,\text{succ\_region}(\joshua{l,}r),k,n}\right)\, ,
\end{align}
where $\Gamma(\cdot)$ is the step function.
\joshua{This binary aggregation only indicates whether workflow deployment $n$ is relevant for the current subproblem through external requests. The corresponding request volume is not represented by $\Gamma(\cdot)$, but by the aggregated request rates introduced below.}
This coupling reduces degrees of freedom compared with centralized optimization and typically increases overall cost, particularly regarding selection-related terms, since regional uniformity restricts flexibility among users sharing similar locations or providers.
Therefore, successful decomposition requires balancing computational feasibility against model accuracy, specifically, there exists a trade-off between reduced computational demand achieved via coarse regional aggregation (solving fewer large-scale problems) and finer granularity necessary to represent realistic distributions among edge-node clusters.
During the solving process, the user requests have to be adapted as well, depending on the decisions made in the upper level. As the head node serves as a placeholder for requests coming from outside the region, we have to additionally add the user request rate depending on the requests from the other regions, that is 
\joshua{Here, $\text{succ\_node}(\joshua{l,}0)$ denotes the upper-level representative of the current regional subproblem, i.e., the node at level $(l-1)$ that corresponds to the local head node indexed by zero.}
\begin{align}
	\label{eq:add_cons_2}
	\lambda_{l,0}=\lambda_{l,0}+\sum_{k}\lambda_{l-1,k}-\lambda_{l-1,\text{succ\_node}(\joshua{l,}0)}\, ,
\end{align}
where we sum over all other regions in the upper level to gather the requests.
With the addition of~\eqref{eq:add_cons_1}-\eqref{eq:add_cons_2}, the main problem \eqref{eq:main} can now be solved for each region $r$ in each level $l$, requiring only communication down the levels during the placement and an update of the current cloud attributes after the placement on the lowest physical node level.
\joshua{\subsection{Scalability and Suboptimality}
	The theoretical properties of the proposed decomposition can be analyzed mathematically. We first provide a scalability analysis and then discuss the suboptimality introduced by the decomposition.
	For scalability, we analyze how the overall computation time $T_{\text{comp}}$ required to solve the decomposed problem grows with the number of nodes $N_s$.
	To simplify the analysis, we make the following assumptions:
	
	\begin{assumption}
		\label{ass:regions}
		The overall number of $N_s$ nodes is divided into regions of size $n_s\in [n_{s,\text{min}},\, n_{s,\text{max}}]$. Likewise, each region in level $l$ is represented by one node in level $l-1$, and a number of $n_s$ regions are grouped into one region in level $l-1$.
	\end{assumption}
	\begin{assumption}
		\label{ass:fixed}
		For a given decomposition, we assume that the computational demand of each subproblem depends only on the number of nodes in the corresponding region. Hence, the computation time of the subproblems is bounded by $T(n_{s,\text{min}})\leq T(n_s)\leq T(n_{s,\text{max}})$.
	\end{assumption}
	\begin{assumption}
		\label{ass:parallel}
		All subproblems within the same level can be solved in parallel, since the algorithm is decentralized and each subproblem is solved by the corresponding regional head node. However, the subproblems in level $l$ require the solution of the corresponding subproblem in level $l-1$. Therefore, the levels have to be processed sequentially.
	\end{assumption}
	These assumptions are valid in the case of large cloud sizes with multiple providers as examined in this work.
	Under these assumptions, the following theorem holds:
	\begin{theorem}
		The computation time of the decomposed approach grows logarithmically with the number of nodes, that is,
		$T_{\text{comp}}(N_s)\sim\mathcal{O}(\lceil \log_{n_{s,\text{min}}}(N_s)\rceil)$.
	\end{theorem}
	
	\begin{proof}
		From Assumption~\ref{ass:regions}, it follows that $L$ levels are required to represent $N_s$ nodes with subregions of size $n_s$ or smaller. This number can be overapproximated by $L\leq\lceil \log_{n_{s,\text{min}}}(N_s)\rceil$. Under Assumption~\ref{ass:parallel}, the time $T_l$ required to solve all subproblems in level $l$ is given by the maximum computation time of any subproblem in that level. Therefore, using Assumption~\ref{ass:fixed}, it follows that $T_l\leq T(n_{s,\text{max}})$.
		Since the levels have to be processed sequentially according to Assumption~\ref{ass:parallel}, the overall computation time $T_{\text{comp}}(N_s)$ is the sum of all level-wise computation times $T_l$ for $l\in[0,\ldots,L-1]$.
		Thus, we obtain
		\begin{align*}
			T_{\text{comp}}(N_s)&=\sum_{l=0}^{L-1}T_l
			\leq\sum_{l=0}^{L-1}T(n_{s,\text{max}})
			\\&=L\cdot T(n_{s,\text{max}})
			\leq \lceil \log_{n_{s,\text{min}}}(N_s)\rceil T(n_{s,\text{max}}).
		\end{align*}
		This proves that $T_{\text{comp}}(N_s)\sim\mathcal{O}(\lceil \log_{n_{s,\text{min}}}(N_s)\rceil)$.
	\end{proof}
	
	In addition to scalability, the suboptimality introduced by the decomposition and by the fixed selection policy has to be considered. The decomposed formulation restricts the feasible set of the centralized problem, since placement and selection decisions are made region-wise and information between regions is exchanged only through the hierarchical head-node structure. In particular, the selection variables $H_{k,n}$ are fixed during placement and cannot adapt to the realized stochastic request sequence during runtime. Moreover, in the decomposed formulation, users aggregated within the same region may be forced to share selection decisions that could be chosen independently in the centralized formulation. Therefore, the decomposed solution is generally suboptimal with respect to the centralized formulation.
	
	A tight analytical bound on this suboptimality is difficult to obtain, because the realized cost depends on the stochastic request sequence, the resulting resource contention, and the interaction between placement and selection decisions. Consequently, we do not claim an a priori optimality certificate for the decomposed solution. Instead, the practical impact of this suboptimality is quantified empirically in the case study by comparing the decomposed formulation with the centralized formulation on a small example and with a heuristic baseline for larger architectures.
}

\section{Case study}

\subsection{\joshua{ Evaluation Methodology}}
In the case study, we analyze the benefits of the proposed algorithm using simulated serverless architectures and synthetic workflows. For modeling purposes, we define parameters describing both workflows and the architecture, such as evaluation costs on cloud nodes, the number of nodes within the serverless architecture, and the number of functions contained in each workflow.
\joshua{Modeling parameters are chosen closely to values of real applications, such as in \cite{GHORBIAN2024103291}.}
\begin{table}[htbp]
	\centering
	\caption{Fixed model parameters for the serverless architecture and workflows.}
	\label{tab:param_fixed}
	\begin{tabular}{lclc}
		\hline
		Architecture &Value& Workflow &Value\\
		\hline
		$p_{\text{RAM}}$ provider $1$&$\SI{0.0167}{\$ MB^{-1}}$&$n_{\text{deployments}}$&[2\, 5]\\
		$p_{\text{send}}$ provider $1$&$\SI{0.01}{\$ MB^{-1}}$&$\text{RAM}$&$[1\, 100]\SI{}{MB s^{-1}}$\\
		$p_{\text{data}}$ provider $1$&$\SI{0.01}{\$ MB^{-1}}$&$t_{\text{func}}$&$[100,\, 1000]\SI{}{s}$\\
		$p_{\text{RAM}}$ provider $2$&$\SI{0.02}{\$ MB^{-1}}$&$x_{\text{func}}$&$[\SI{0.75}{},\, \SI{1}{}]$\\
		$p_{\text{send}}$ provider $2$&$5\cdot 10^{-3}\SI{}{\$ MB^{-1}}$&$d_{\text{\joshua{in}}}$&$[0,\, 100]\SI{}{MB}$\\
		$p_{\text{data}}$ provider $2$&$1\cdot 10^{-3}\SI{}{\$ MB^{-1}}$&$d_{\text{send}}$&$[100,\, 1000]\SI{}{MB}$\\
		$L_{\text{nodes}}$&$[0,\,10]\SI{}{s}$&$d_{\text{data}}$&$[0,\, 1000]\SI{}{MB}$\\
		$b_{\text{decrease}}$&$\SI{10}{}$&$\rho_{\text{data}}$&$\SI{0.2}{}$\\
		$n_{l}$&2&$n_{\text{func}}$&$[3,\, 6]$\\
		$\lambda_k$&$[0.03\, 0.05]$&&\\
		\hline
	\end{tabular}
\end{table}
These parameters are either uniformly distributed within fixed bounds or remain constant throughout the case study. Table~\ref{tab:param_fixed} defines these parameters together with the values used across all experiments.
The serverless architecture is modeled with two cloud providers and one type of edge node. Each provider has fixed costs for sending data, executing functions, and transferring dependent data, chosen to approximate realistic values for commercial cloud services~\cite{xie_workflow_2023}. Additional architectural parameters include bounds for latencies and speedups on nodes. To achieve a more realistic simulation, latency bounds are reduced in each level to reflect increasingly regional connections according to the following formula
\begin{align*}
	L_{l,\text{upper bound}}=L_{\text{nodes}} b_{\text{decrease}}^{-l}\, .
\end{align*}
Function attributes are randomly sampled to generate diverse workflows. These attributes include evaluation time requirements, data sizes exchanged between functions, dependency data sizes, and the probability that a function depends on external data. In addition to these workflow parameters, we define several architectural scaling parameters that influence problem size.
The number of nodes per level is uniformly distributed, combined with the total number of levels, these two parameters determine the overall scale of the serverless architecture. Furthermore, we separately specify the number of nodes considered in the highest level. For workflows, scalability is controlled by varying both the total number of workflows deployed on the architecture and the number of functions contained within each workflow.
As the actual costs for a placed workflow can only be evaluated in the online selection phase, the case study simulates requests in a simulation time of $t_{\text{sim}}=\SI{200}{s}$. The requests are randomly sampled with a Poisson distribution using the request rate $\lambda_k$ for each physical user $k$. These requests are scheduled using a FIFO scheduling, where for each workflow request, the future resources are blocked immediately. In this simulation, we therefore model two blocking mechanisms: A specific function $m$ in a workflow $t$ with the deployment number $n$ cannot be used twice in the same time step, and if the RAM resource of node $i$ is fully occupied, no further function can be executed on said node.

Instead of the heuristic utilization cost $U$ that is used in the placement problem, we focus, therefore, on the actual waiting time $W$ due to blocking.
We can therefore formulate the overall cost $C$, defined as a weighted sum of the two competing objectives,
\begin{align*}
	C = w_1 M + w_2 (B + W)\, .
\end{align*}
In this study, we additionally will have to show that the weight $w_3$ has an influence on reducing the waiting time.
\joshua{The stochastic simulation emulates realistic workflow request scenarios and is therefore essential for evaluating how effectively the proposed placement formulation, including its proxy utilization cost, performs under stochastic request arrivals. Instead of relying on real workflow traces, this work considers synthetic workflows whose attributes, such as function execution times, data demands, and resource allocations, are initially assumed to match the estimates available during deployment. Similar simulation-based evaluations are widely used in the literature \cite{costless,tarot,neptune,auctionwhisk,luo2024efficient,MOAKHAR2024103890}, where the required parameters are typically obtained through an initial profiling or monitoring phase. Extending the evaluation to real traces requires accounting for deviations between estimated and realized parameters and for the resulting suboptimality of the placement. This aspect is investigated in the final subsection of this case study.}
For comparison reasons, we frequently evaluate the gap between the cost of two different solution strategies
\begin{align*}
	g=\frac{C_1-C_2}{C_1}\, ,
\end{align*}
with $C_1>C_2$.
Second, we analyze the computation time required to solve the placement problem. In the centralized formulation, this corresponds to solving a single global optimization problem. In contrast, under decomposition we assume that total computation time equals the sum over maximum subproblem times at each level,
\begin{align*}
	t_c = \sum_{l=0}^{N_{\text{level}}-1}\underset{r}{\operatorname{max}}\, t_{c,r,l}\, ,
\end{align*}
where $t_{c,r,l}$ denotes computation time for subproblem~$r$ at level~$l$. Since subproblems within a given level can be solved in parallel, only their maximum runtime contributes per level. The communication time required to pass information between levels is neglected here, as the centralized approach would also need constant communication between the nodes to update the central parameters and because we want to focus on the computational burden of the placement optimization problem in itself. 
All simulations are performed on an AMD Ryzen 9 16‑Core processor equipped with $\SI{94}{GB}$ RAM. The modeling framework and cost calculations are implemented in Python. Integer decision problems are solved using Gurobi~\cite{gurobi}. The code necessary to generate the presented results is openly available\footnote{\scriptsize{\url{https://github.com/JoshuaAda/IEA_LSWPISCUINP}}}. 
\subsection{Small working example}
\label{sub:small_example}
First, we reintroduce the small example used in the modeling section, now with real parameters as described in Table~\ref{tab:param_fixed} to illustrate the modeling and solution structure of the proposed placement algorithm and highlight its differences compared to the centralized approach. For ease of visualization, we reduce the number of deployments to two, the number of functions for each workflow to three, the number of nodes in the second level to three, and the number of edge node regions in the first level to one. 
Figure~\ref{fig:results_centralized} shows the full problem, considering all nodes in one optimization problem.
In comparison, Figure~\ref{fig:results_decomposed} shows the serverless architecture considered in the decomposed setting, consisting of two hierarchical levels. 

\begin{figure}[htbp]
	\centering
	\includegraphics[width=0.42\textwidth]{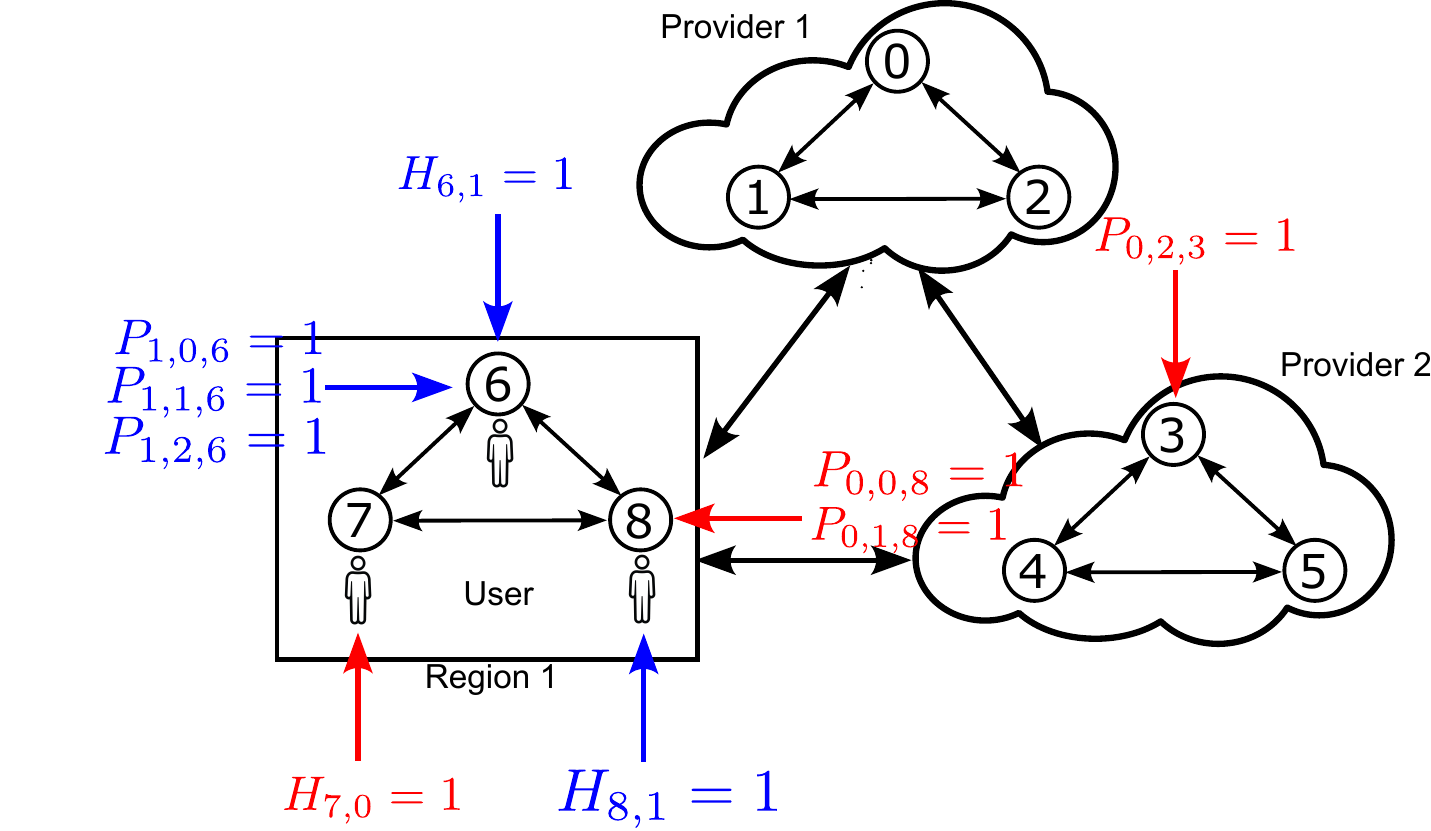}
	\caption{\joshua{Placement solution of the centralized approach for two workflows (first workflow  placements and selections in red, second in blue) containing three functions. In the centralized approach, the three different users can select different workflows (user 6 and 8 selecting the \joshua{second}, user 7 the \joshua{first}), giving more freedom to reduce blocking as well as basic time and monetary costs.}}
	\label{fig:results_centralized}
\end{figure}
\begin{figure}[htbp]
	\centering
	\includegraphics[width=0.42\textwidth]{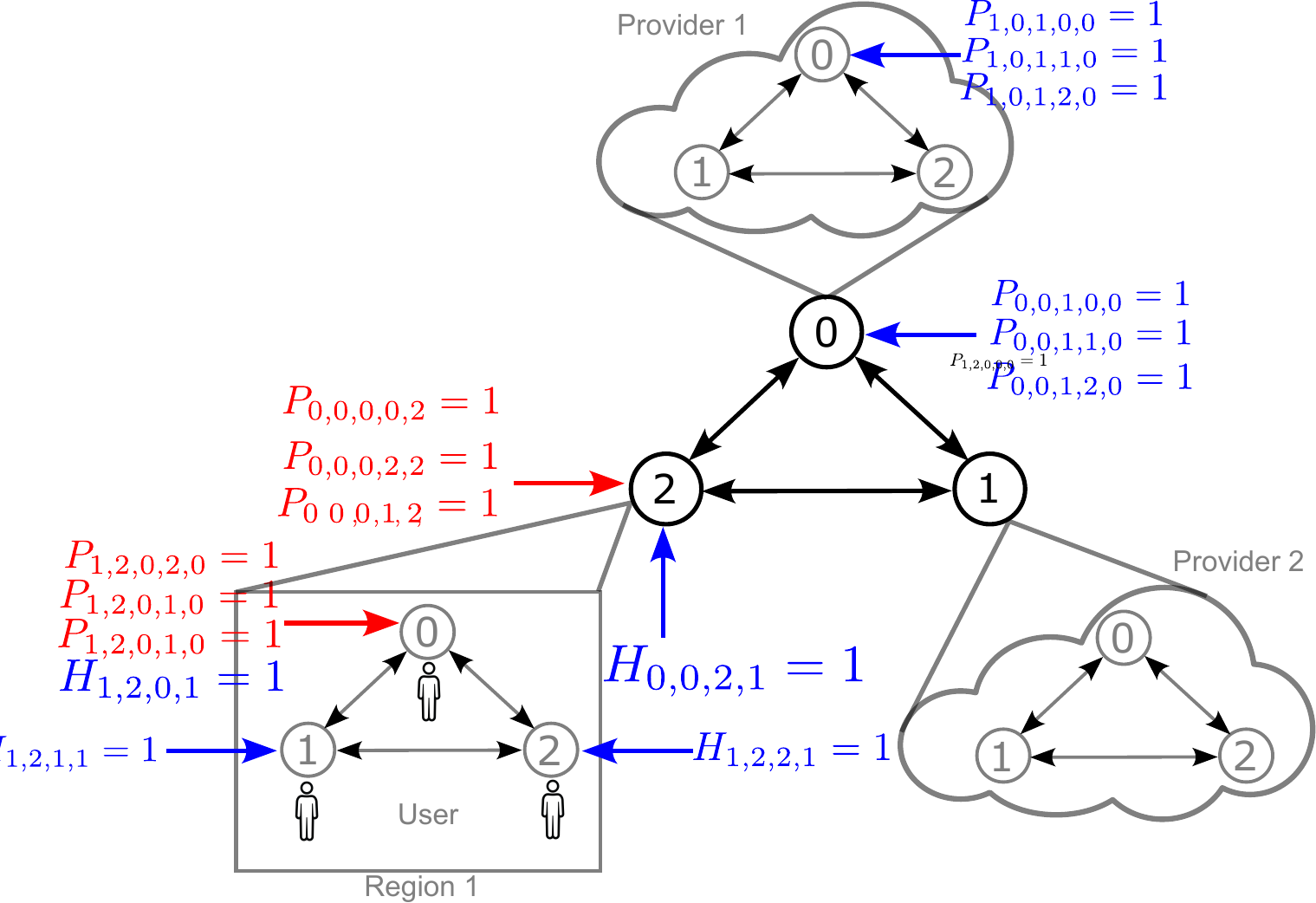}
	\caption{\joshua{Placement solution of the decomposed approach for two workflows (first workflow  placements and selections in red, second in blue) containing three functions. In the decomposed approach, only one workflow can be selected as the three users are lumped in one region on the first level of optimization. The functions are placed on the regions in the first level and are later distributed onto the nodes of the regions in the second level.}}
	\label{fig:results_decomposed}
	
\end{figure}

In the centralized case, three users can send requests independently, whereas in the decomposed approach, only one workflow can be selected per region, resulting in slightly higher overall cost.
In total, we solve one optimization problem containing $131$ decision variables for the centralized formulation and four subproblems for the decomposed approach. The main problem in the decomposition approach includes $49$ decision variables, while the sizes of subordinate problems depend on decisions made at this top level ($26,\,$\joshua{$0$}, and $26$ decision variables for the specific subproblems).
Figure~\ref{fig:results_centralized} presents results for the centralized solution, while Figure~\ref{fig:results_decomposed} shows outcomes for the decomposed strategy assuming equal weights for the monetary cost and time cost $w_1 = w_2 = 0.5$, and a utilization cost of $w_3=\joshua{1}$. The primary difference between the centralized and the decomposed approach lies in the selection behavior. 
This difference can be seen in the small example. In Figure~\ref{fig:results_decomposed}, all three users can select the \joshua{second} workflow, whereas the \joshua{first} workflow has no influence on the cost. On the contrary, the centralized approach can select different workflows for different users, such that the second workflow is selected for the second user. The result follows directly from the costs presented in~\eqref{eq:main}. In general, it is beneficial to place the functions on the regional edge nodes as they offer small latencies and free evaluation. As the first function in the workflow can be completed much faster on the cloud node of the first provider, it is placed there. As the latencies between the regional nodes are very small, the selection cost leads to a more equal distribution of the function onto the edge nodes.
The total cost $C$ across all three users in the simulation amounts to \joshua{$305.27$} for the centralized approach and \joshua{$427.07$} for the decomposed variant. The main difference stems from the waiting time, which is $\SI{137}{s}$ in the decomposed case against $\SI{79}{s}$ in the centralized approach, which is due to the higher degree of freedom to choose the workflows in the centralized approach. Despite the higher cost due to restricted selection flexibility, computation time decreases even within this small-scale scenario substantially from \joshua{$\SI{217.93}{ms}$} for centralized optimization to \joshua{$\SI{47.57}{ms}$} using the decomposition.

\subsection{Scaling of the decomposition approach}
\label{sub:scaling}
After demonstrating the main functionalities of the proposed algorithm, we now highlight its principal advantage: the significantly improved scalability of the algorithm compared to the centralized approach. This aspect is crucial for realistic workflow placement scenarios, as the number of nodes expected in practical serverless architectures is typically significantly larger than those considered in the previous subsection.
To scale the architecture, two parameters can be adjusted. First, the number of nodes within each subregion can be increased either as a fixed value or as bounds for a uniform distribution from which these values are sampled, which leads to a moderate increase in total node count. Second, we can increase the number of hierarchical levels used to organize nodes within the architecture. This results in an exponential growth of total nodes and provides a useful setting for evaluating scalability under large system sizes.
Due to this exponential growth in decision variables with increasing numbers of nodes, we expect that computation time for the centralized approach will scale poorly. Therefore, we compare both approaches using an architecture with two levels while gradually increasing node counts within subregions. Other workflow and node parameters are sampled uniformly within bounds provided in Table~\ref{tab:param_fixed}. Table~\ref{tab:scaling_comp} shows overall costs and computation times for both approaches. Costs are slightly higher for the decomposed method across all configurations, the maximum gap observed between both approaches is \joshua{$\SI{5.23}{\%}$}. 
\begin{table}[htbp]
	\centering
	\caption{Overall cost $C$ and computation time $t_c$ for a two-level decomposition with an increasing number of nodes $n_{\text{nodes},1,r}$ in the second level for every region in comparison to the centralized approach that solves the full optimization problem.}
	\label{tab:scaling_comp}
	\begin{tabular}{lccc}
		\hline
		Value &$n_{\text{nodes},1,\joshua{r}}=2$&$n_{\text{nodes},1,\joshua{r}}=3$&$n_{\text{nodes},1,\joshua{r}}=4$\\
		\hline
		$C$  [-] (decomposed)&$\joshua{16510.59}$&\joshua{$13799.34$}&\joshua{$15379.66$}\\
		$C$ [-] (centralized)&$\joshua{\mathbf{15690.41}}$&$\joshua{\mathbf{13545.68}}$&$\joshua{\mathbf{15024.66}}$\\
		$g[\%]$&\joshua{$5.23$}&\joshua{$1.87$}&\joshua{$2.36$}\\
		$t_c$ (decomposed) [s]&$\joshua{\mathbf{1.65}}$&$\joshua{\mathbf{1.90}}$&$\joshua{\mathbf{2.03}}$\\
		$t_c$ (centralized) [s]&\joshua{$12.74$}&\joshua{$67.35$}&\joshua{$161.13$}\\
		\hline
	\end{tabular}
\end{table}

For the centralized optimization, the runtime increases exponentially with the number of nodes. For instance, placing two workflows on a serverless architecture containing a total of $16$ nodes already requires \joshua{$\SI{161.13}{s}$}. 

In contrast, computation time for the decomposed approach increases logarithmically with node count (linear with respect to the number of levels) due to parallel evaluation of subproblems at each level.

These findings confirm that decomposition enables efficient scaling toward large architectures while providing accurate approximations to global optimal solutions.
There is no clear increase or decrease of the gap between the centralized and decomposed solution with the rising number of nodes, confirming that the approximation error is not dependent from the number of nodes. The gap instead depends highly on the chosen scenario. For example, a rising number of deployments can allow for a much bigger gap due to the additional degree of freedom to choose workflows for the centralized approach as discussed in Subsection~\ref{sub:small_example}. For the computation time, we achieve a favorable scaling using the decomposed approach in contrast to the exponential growth in the centralized approach.

To further validate scalability benefits, we next increase architectural depth by varying hierarchical levels, thereby enabling substantially larger physical number of nodes. In this scenario, direct comparison against centralized optimization becomes infeasible since it fails to produce solutions. Instead, we benchmark against a simple heuristic strategy that places all functions on one cloud provider (provider 1). \joshua{This cloud-only placement serves as a standard baseline in related work~\cite{xie_workflow_2023}. We further note that other heuristic-based approaches, such as \cite{DEHURY2024111179,smith}, address different variants of the serverless placement problem with different assumptions and objectives. Therefore, they are not directly comparable to the decentralized workflow placement setting considered in this work.}
In our case, it serves as a good heuristic, as blocking because of RAM constraints is reduced as much as possible with this heuristic.
In this experiment, we scale from two up to four levels while assuming two regional edge nodes (and again two cloud nodes) initially. Each subproblem contains ten physical nodes per region or provider, consequently total node counts reach $40$, $400$, and $4000$ physical devices at successive levels across three test cases.
\begin{table}[htbp]
	\centering
	\caption{Overall cost $C$ and computation time $t_c$ for increased number of levels $n_{\text{l}}$ and therefore the total number of nodes in the last level $n_{\text{nodes}}=\sum_r n_{\text{nodes},n_l-1,r}$ in the decomposition in comparison to a cloud only heuristic.}
	\label{tab:scaling_heuristic}
	\begin{tabular}{lccc}
		\hline
		Value &\makecell{$n_{\text{l}}=2$\\$n_{\text{nodes}}=40$}&\makecell{$n_{\text{l}}=3$\\$n_{\text{nodes}}=400$}&\makecell{$n_{\text{l}}=4$\\$n_{\text{nodes}}=4000$}\\
		\hline
		$C$  [-] (decomposed)&\joshua{$\mathbf{13873.43}$}&\joshua{$\mathbf{11294.04}$}&\joshua{$\mathbf{13591.81}$}\\
		$C$ [-] (heuristic)&\joshua{${16942.93}$}&\joshua{${12141.04}$}&\joshua{${13863.40}$}\\
		$g[\%]$&\joshua{$22.12$}&\joshua{$7.50$}&\joshua{$1.99$}\\
		$t_c$ (decomposed) [s]&\joshua{${4.85}$}&\joshua{${10.59}$}&\joshua{${37.37}$}\\
		$t_c$ (heuristic) [s]&$\mathbf{0.00}$&$\mathbf{0.00}$&$\mathbf{0.00}$\\
		\hline
	\end{tabular}
	
\end{table}
Table~\ref{tab:scaling_heuristic} shows the cost and computation times for the decomposed approach and a cloud-only heuristic for an increased number of levels. It shows that the proposed decomposed approach consistently outperforms this heuristic regarding overall cost. Computation times reveal negligible runtime requirements for heuristics but linear growth for decomposition. This is an expected outcome since additional levels introduce more parallelizable subproblems rather than sequential dependencies.

Overall findings demonstrate that employing optimal placement through integer programming becomes feasible even at large scales when combined with practical decomposition strategies and has significant advantages over simple heuristics rules. The proposed method thus compromised between theoretical optimality and computational tractability required for realistic serverless environments.

\subsection{Sensitivity to weight parameters}
\begin{figure}[b]
	
	\centering
	\begin{subfigure}[t]{.21\textwidth}
		\centering
		\includegraphics[width=\textwidth]{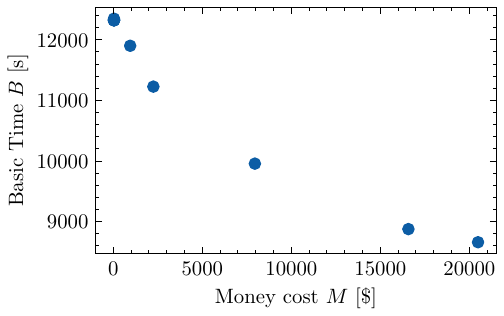}
		\caption{\joshua{$w_3=0.$}}
		\label{fig:pareto_0}
	\end{subfigure}
	\begin{subfigure}[t]{.21\textwidth}
		\centering
		\includegraphics[width=\textwidth]{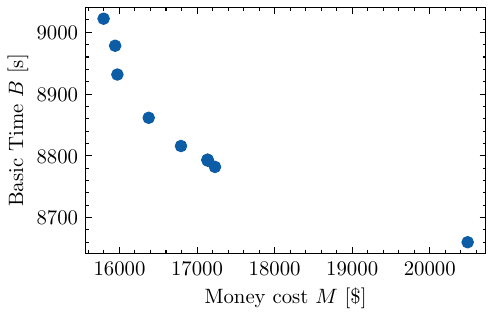}
		\caption{\joshua{$w_3=3.$}}
		\label{fig:pareto_3}
	\end{subfigure}\\
	\begin{subfigure}[t]{.21\textwidth}
		\centering
		\includegraphics[width=\textwidth]{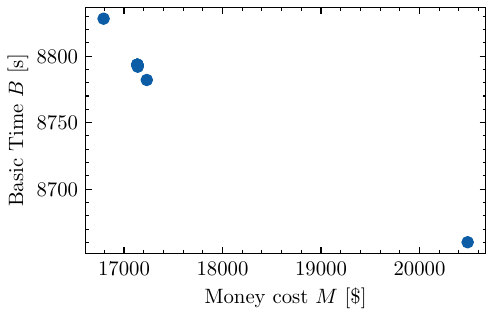}
		\caption{\joshua{$w_3=6.$}}
		\label{fig:pareto_6}
	\end{subfigure}
	\begin{subfigure}[t]{.21\textwidth}
		\centering
		\includegraphics[width=\textwidth]{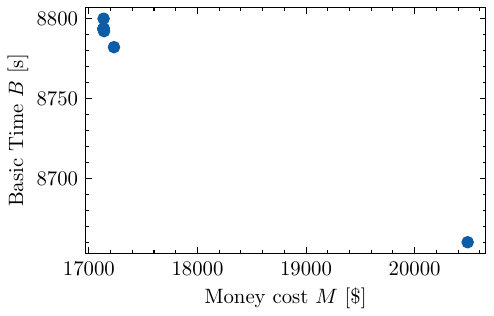}
		\caption{\joshua{$w_3=9.$}}
		\label{fig:pareto_9}
	\end{subfigure}
	\caption{Pareto fronts for monetary and basic time costs using a weighted sum approach (increase weight $w_1$ for basic time cost, while decrease weight $w_2=1-w_1$ for the monetary cost), exemplary for different values of penalization of utilization $w_3$.}
	\label{fig:paretos}
\end{figure}
Until this subsection, we have evaluated the solution of the placement problem with fixed weights. As described in~\eqref{eq:main}, we consider two competing costs (monetary and time) with the additional utilization cost that in turn reduces the waiting time and therefore again the overall time to run a workflow.
Next, we want to verify that monetary and time costs are competing and that the algorithm can be tuned using $w_1$ and $w_2$ to reduce or increase the monetary and basic time costs. Furthermore, we investigate whether increasing the utilization cost factor $w_3$ leads to a reduced waiting time and whether it has any impact on the tuning of the basic time and monetary cost.
Understanding how varying the relative importance of each cost factor influences overall performance is essential when deciding which objective should receive greater emphasis.
Such analyses are commonly conducted through Pareto front evaluation~\cite{ngatchou_pareto_2005}.
The Pareto front represents all points in the multi-dimensional cost space for which improvement in one objective necessarily leads to deterioration in at least one other. As proposed in~\cite{marler_weighted_2010}, the optimal solution obtained from a weighted-sum formulation, such as the placement problem defined in~\eqref{eq:main}, corresponds to a point located on this Pareto front.
To further investigate trade-offs among monetary cost and evaluation time of the workflow, we perform a simulation involving ten workflows deployed on a serverless architecture with attributes given in Table~\ref{tab:param_fixed}. This experiment is conducted for different weights between the two competing costs (eleven values between zero and one for $w_1$ and $w_2=1-w_1$). Additionally, this process is repeated for eleven additional values of $w_3$ between zero and ten.

\begin{table}[htbp]
	
	\centering
	\caption{Overall waiting time in seconds evaluated with increased penalization of utilization $w_3$ for different weight combinations $w_1/w_2$ on the Pareto front.}
	\label{tab:waitings}
	\begin{tabular}{lcccc}
		\hline
		&$w_3=0$&$\joshua{w_3=1}$&$\joshua{w_3=2}$&$\joshua{w_3=3}$\\
		\hline
		$w_1=0.0; w_2=1.0$&\joshua{$2246$}&\joshua{$1568$}&\joshua{$1568$}&\joshua{$1568$}\\
		$w_1=0.3; w_2=0.7$&\joshua{$283852$}&\joshua{$1587$}&\joshua{$1568$}&\joshua{$1568$}\\
		$w_1=0.6; w_2=0.4$&\joshua{$384092$}&\joshua{$2252$}&\joshua{$1742$}&\joshua{$1582$}\\
		$w_1=0.9; w_2=0.1$&\joshua{$232362$}&\joshua{$16224$}&\joshua{$1817$}&\joshua{$1804$}\\
		\hline
	\end{tabular}
	
\end{table}
Figure~\ref{fig:paretos} shows four exemplary Pareto fronts for basic time cost $B$ against the monetary cost $M$. Each example is done for a specific value of $w_3$ and is generated by increasing the weight for basic time cost $w_1$, while reducing the weight of the monetary cost $w_2$. The Pareto fronts show that both cost factors are competing and can be tuned with $w_1$ and $w_2$ for all values of $w_3$. Additionally, Table~\ref{tab:waitings} shows the reduced waiting time with increasing $w_3$ for different combinations of $w_1$ and $w_2$ \joshua{for the lower value experiments of $w_3$}. Without considering the utilization cost for all combinations of $w_1$ and $w_2$, the waiting time is much longer than the basic time $T$. Even for small values of $w_3$ the waiting time reduces fast and remains stable, demonstrating that the utilization penalty effectively reduces waiting time due to blocking in simulation.
\joshua{\subsection{Sensitivity to incorrect model parameters and real world applicability}
	Since the proposed algorithm is evaluated in a high-fidelity simulation and not directly on real workflow deployments, we next discuss the real-world applicability of the approach. The main bottleneck for practical deployment is the uncertainty of essential model parameters, such as request rates, execution times, latencies, and resource demands. Many approaches estimate these values in an initial profiling phase before deployment \cite{GHORBIAN2024103291}. However, these parameters may change during runtime or may be difficult to estimate accurately from short observation windows.
	
	We therefore investigate scenarios in which the parameters used in the simulation differ from the parameters assumed in the placement optimization. The placement remains feasible under such deviations, but the resulting solution may become suboptimal. The investigation is conducted on an architecture with two levels and a total of $100$ physical nodes, as used in the comparison against the heuristic in Subsection~\ref{sub:scaling}, with weights chosen as $w_1=0.5$, $w_2=0.5$, $w_3=1$. To analyze the influence of incorrect parameter estimates, we scale one uncertain parameter at a time in the simulation while keeping the value used in the placement problem fixed. We consider four uncertain parameters: the latencies $L$, the user request rates $\lambda$, the allocated RAM, and the runtime of each function.
	
	In a first step, we analyze the sensitivity of the simulated cost with respect to these parameter deviations.
	\begin{figure}[htbp]
		\centering
		\includegraphics[width=0.48\textwidth]{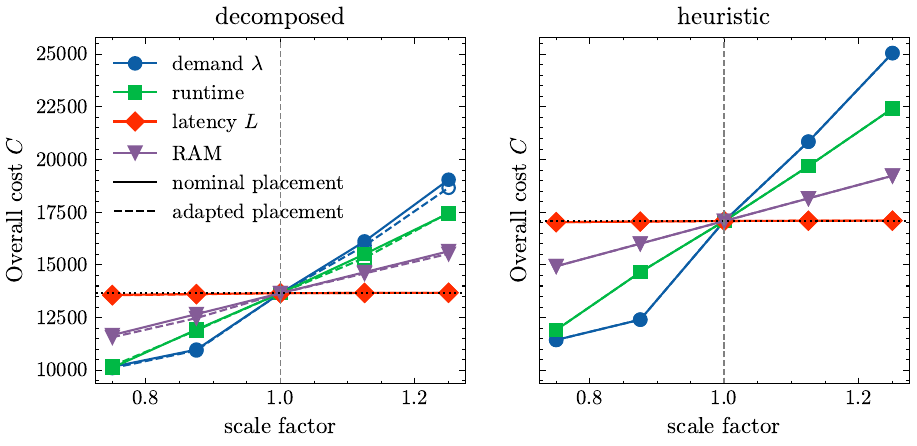}
		\caption{\joshua{Sensitivity analysis for the simulated placement costs for four uncertain parameters: latencies, allocated RAM, function runtimes, and user request rates. Each parameter is scaled in the simulation relative to the value assumed in the placement problem. The comparison is performed against the heuristic introduced in Subsection~\ref{sub:scaling}.}}
		\label{fig:sensitivity}
	\end{figure}
	Figure~\ref{fig:sensitivity} shows that the simulated cost increases monotonically with the scaling factor for all investigated parameters. Thus, if the realized parameters are smaller than assumed, the simulated cost decreases, whereas larger realized values lead to higher costs. This behavior is observed for both the decomposed approach and the heuristic baseline. The monotonic response is important for practical deployment, since it indicates that conservative over-approximations of uncertain parameters can be used to obtain more robust placements against worst-case parameter deviations.
	
	In a second step, we compare the nominal placement, which is computed using the original parameter estimates, with an adapted placement that is recomputed using the true scaled parameter values.
	\begin{figure}[htbp]
		\centering
		\includegraphics[width=0.45\textwidth]{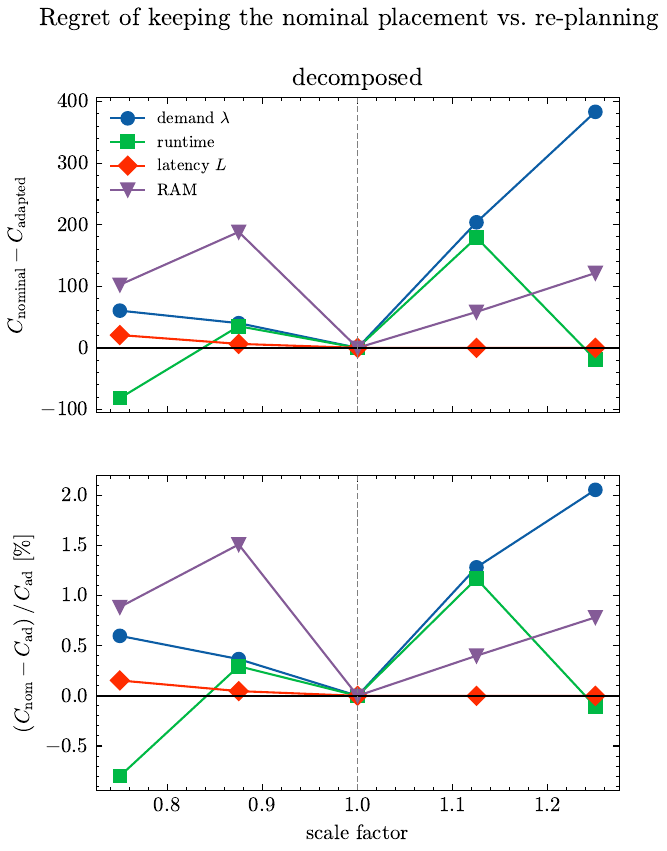}
		\caption{\joshua{Regret of keeping the nominal placement compared with re-planning using the true scaled parameter values. Positive values indicate that re-planning improves the simulated cost.}}
		\label{fig:sensitivity_replanning}
	\end{figure}
	Figure~\ref{fig:sensitivity_replanning} shows that the benefit of re-planning is comparatively small and does not exhibit a clear monotonic trend with the scaling factor. In some cases, adapting the placement to the true parameter values improves the simulated cost, while in other cases the difference is negligible. This behavior is expected, since the utilization term used in the placement problem is a heuristic proxy for the waiting time observed in the stochastic simulation. While the parameters have a significant influence on the cost, there is little than can be done in the optimization to prevent this from happening. 
	Overall, this study show that cost rises monotone with possible uncertain parameters, allowing for worst-case approximations in real world applications to enable robust solutions against parametric uncertainty. Furthermore, the study shows that the difference between choosing the perfect parameters and a deviated estimate is small, further pointing toward real world applicability of the approach.
}

\section{Conclusion}
In this work, we present a modeling framework for the optimal placement of workflows in serverless fog computing architectures. The framework jointly models workflow placement and selection within the architecture and formulates a nonlinear optimization problem that minimizes both monetary cost and execution time. To address decomposed information and the exponential growth of computational demand with increasing numbers of nodes, we introduce a decomposition strategy enabling scalable optimization.
We demonstrate the advantages of the proposed method compared to a centralized formulation and heuristic baseline, as well as its adaptability through a Pareto analysis of competing cost objectives. 
Future work will extend the framework toward \joshua{advanced} probabilistic online selection strategies and incorporate additional cost factors relevant to practical deployment scenarios, such as workflow cold-start delays and energy consumption of serverless nodes.
\joshua{Additionally, real world application requires adaption to incorrect parameter estimation. Future work will address this challenge through recourse action in the placement problem as well as robust optimization problem formulations that allows to deal with uncertain parameter settings.}



\section*{Declaration of generative AI and AI-assisted technologies in the manuscript preparation process}
During the preparation of this work the authors used GPT 5.0-5.\joshua{5} in order to make adjustments to existing text in this paper, as well as small adaptation to coding sections for data storage and figure creation. After using this tool/service, the author(s) reviewed and edited the content as needed and take(s) full responsibility for the content of the published article.
\printbibliography

@article{MOAKHAR2024103890,
title = {An efficient mechanism for function scheduling and placement in function as a service edge environment},
journal = {Journal of Network and Computer Applications},
volume = {226},
pages = {103890},
year = {2024},
}

@article{luo2024efficient,
  title={Efficient and flexible component placement for serverless computing},
  author={Luo, Shouxi and Li, Ke and Xing, Huanlai and Fan, Pingzhi},
  journal={IEEE Systems Journal},
  volume={18},
  number={2},
  pages={1104--1114},
  year={2024},
  publisher={IEEE}
}

@article{neptune,
author = {Baresi, Luciano and Hu, Davide Yi Xian and Quattrocchi, Giovanni and Terracciano, Luca},
title = {NEPTUNE: A Comprehensive Framework for Managing Serverless Functions at the Edge},
year = {2024},
issue_date = {March 2024},
publisher = {Association for Computing Machinery},
address = {New York, NY, USA},
volume = {19},
number = {1},
journal = {ACM Trans. Auton. Adapt. Syst.},
}

@inproceedings{tarot,
  title={TAROT: spatio-temporal function placement for serverless smart city applications},
  author={De Maio, Vincenzo and Bermbach, David and Brandic, Ivona},
  booktitle={2022 IEEE/ACM 15th International Conference on Utility and Cloud Computing (UCC)},
  pages={21--30},
  year={2022},
  organization={IEEE}
}

@article{auctionwhisk,
author = {Bermbach, David and Bader, Jonathan and Hasenburg, Jonathan and Pfandzelter, Tobias and Thamsen, Lauritz},
title = {AuctionWhisk: Using an auction-inspired approach for function placement in serverless fog platforms},
journal = {Software: Practice and Experience},
volume = {52},
number = {5},
pages = {1143-1169},
year = {2022}
}

@INPROCEEDINGS{costless,
  author={Elgamal, Tarek and Sandur, Atul and Nahrstedt, Klara and Agha, Gul},
  booktitle={2018 IEEE/ACM Symposium on Edge Computing (SEC)}, 
  title={Costless: Optimizing Cost of Serverless Computing through Function Fusion and Placement}, 
  year={2018},
  volume={},
  number={},
  pages={300-312},}

@INPROCEEDINGS{smith,
  author={Smith, Christopher Peter and Jindal, Anshul and Chadha, Mohak and Gerndt, Michael and Benedict, Shajulin},
  booktitle={2022 IEEE 6th International Conference on Fog and Edge Computing (ICFEC)}, 
  title={FaDO: FaaS Functions and Data Orchestrator for Multiple Serverless Edge-Cloud Clusters}, 
  year={2022},
  volume={},
  number={},
  pages={17-25},}

@article{DEHURY2024111179,
title = {Def-DReL: Towards a sustainable serverless functions deployment strategy for fog-cloud environments using deep reinforcement learning},
journal = {Applied Soft Computing},
volume = {152},
pages = {111179},
year = {2024},
author = {Chinmaya Kumar Dehury and Shivananda Poojara and Satish Narayana Srirama},
}

@article{xu2023stateful,
  title={Stateful serverless application placement in MEC with function and state dependencies},
  author={Xu, Zichuan and Zhou, Lizhen and Liang, Weifa and Xia, Qiufen and Xu, Wenzheng and Ren, Wenhao and Ren, Haozhe and Zhou, Pan},
  journal={IEEE Transactions on Computers},
  volume={72},
  number={9},
  pages={2701--2716},
  year={2023},
  publisher={IEEE}
}

@article{WANG2024242,
title = {SD-SRF: An Intelligent Service Deployment Scheme for Serverless-operated Cloud-Edge Computing in 6G Networks},
journal = {Future Generation Computer Systems},
volume = {151},
pages = {242-259},
year = {2024},
author = {Luying Wang and Anfeng Liu and Neal N. Xiong and Shaobo Zhang and Tian Wang and Mianxiong Dong},
}

@article{raza2023configuration,
  title={Configuration and placement of serverless applications using statistical learning},
  author={Raza, Ali and Akhtar, Nabeel and Isahagian, Vatche and Matta, Ibrahim and Huang, Lei},
  journal={IEEE Transactions on Network and Service Management},
  volume={20},
  number={2},
  pages={1065--1077},
  year={2023},
  publisher={IEEE}
}

@article{GHORBIAN2024103291,
title = {Function Placement Approaches in Serverless Computing: A Survey},
journal = {Journal of Systems Architecture},
volume = {157},
pages = {103291},
year = {2024},
author = {Mohsen Ghorbian and Mostafa Ghobaei-Arani and Rohollah Asadolahpour-Karimi},

}

@inproceedings{alarbi2024eco,
  title={ECO: edge continuum orchestrator framework for managing serverless chains across the cloud-edge spectrum},
  author={Alarbi, Muhamed and Belson, Robert and Lutfiyya, Hanan},
  booktitle={2024 33rd International Conference on Computer Communications and Networks (ICCCN)},
  year={2024},
}

@inproceedings{mahgoub2021sonic,
  title={$\{$SONIC$\}$: Application-aware data passing for chained serverless applications},
  author={Mahgoub, Ashraf and Shankar, Karthick and Mitra, Subrata and Klimovic, Ana and Chaterji, Somali and Bagchi, Saurabh},
  booktitle={2021 USeNIX annual technical conference},
  pages={285--301},
  year={2021}
}

@misc{gurobi,
  author = {{Gurobi Optimization, LLC}},
  title = {{Gurobi Optimizer Reference Manual}},
  year = 2024,
  url = "https://www.gurobi.com"
}

@article{wang2024briskchain,
  title={BriskChain: Decentralized function composition for high-performance serverless computing},
  author={Wang, Kan and Ma, Jing and Lai, Edmund M-K},
  journal={IEEE Access},
  volume={12},
  pages={131313--131322},
  year={2024},
}

@INPROCEEDINGS{ristov_2022,

  author={Ristov, Sashko and Gritsch, Philipp},

  booktitle={2022 IEEE International Conference on Cluster Computing (CLUSTER)}, 

  title={FaaSt: Optimize makespan of serverless workflows in federated commercial FaaS}, 

  year={2022},

  volume={},

  number={},

  pages={183-194},

 }

@ARTICLE{jia,
  author={Jia, Yunshan and Jin, Chao and Li, Qing and Liu, Xuanzhe and Jin, Xin},
  journal={IEEE Transactions on Networking}, 
  title={FaaSPR: Latency-Oriented Placement and Routing Optimization for Serverless Workflow Processing}, 
  year={2025},
  volume={33},
  number={4},
  pages={2063-2078},
}

@inproceedings{gallego2023machine,
  title={Machine learning inference on serverless platforms using model decomposition},
  author={Gallego, Adrien and Odyurt, Uraz and Cheng, Yi and Wang, Yuandou and Zhao, Zhiming},
  booktitle={Proceedings of the IEEE/ACM 16th International Conference on Utility and Cloud Computing},
  pages={1--6},
  year={2023}
}

@ARTICLE{yue,
  author={Yue, Xiaofei and Yang, Song and Li, Fan and Zhu, Liehuang and Wang, Xu and Feng, Zhen and Kuipers, Fernando A.},
  journal={IEEE Transactions on Parallel and Distributed Systems}, 
  title={HyFaaS: Accelerating Serverless Workflows by Unleashing Hybrid Resource Elasticity}, 
  year={2026},
  volume={37},
  number={1},
  pages={272-286},
  }

@ARTICLE{wang,
  author={Wang, Zhijie and Zhang, Zhen and He, Tengjiao and Xie, Hao},
  journal={IEEE Transactions on Network and Service Management}, 
  title={LPCD: A Parallel Candidate Deployment Strategy in Stateful Serverless Computing With Low Latency}, 
  year={2025},
  volume={22},
  number={6},
  pages={5927-5944},
 }

@INPROCEEDINGS{palade,
  author={Palade, Andrei and Mukhopadhyay, Atri and Kazmi, Aqeel and Cabrera, Christian and Nomayo, Evelyn and Iosifidis, Georgios and Ruffini, Marco and Clarke, Siobhán},
  booktitle={2020 IEEE International Conference on Services Computing (SCC)}, 
  title={A Swarm-based Approach for Function Placement in Federated Edges}, 
  year={2020},
  volume={},
  number={},
  pages={48-50},
  }

@INPROCEEDINGS{majewski,
  author={Majewski, Marcin and Pawlik, Maciej and Malawski, Maciej},
  booktitle={2021 IEEE/ACM 21st International Symposium on Cluster, Cloud and Internet Computing (CCGrid)}, 
  title={Algorithms for scheduling scientific workflows on serverless architecture}, 
  year={2021},
  volume={},
  number={},
  pages={782-789},}

@article{xie_workflow_2023,
	title = {Workflow {Scheduling} in {Serverless} {Edge} {Computing} for the {Industrial} {Internet} of {Things}: {A} {Learning} {Approach}},
	volume = {19},
	copyright = {https://ieeexplore.ieee.org/Xplorehelp/downloads/license-information/IEEE.html},
	shorttitle = {Workflow {Scheduling} in {Serverless} {Edge} {Computing} for the {Industrial} {Internet} of {Things}},
	number = {7},
	journal = {IEEE Transactions on Industrial Informatics},
	author = {Xie, Renchao and Gu, Dier and Tang, Qinqin and Huang, Tao and Yu, Fei Richard},
	month = jul,
	year = {2023},
	pages = {8242--8252},
}

@article{eismann_state_2022,
	title = {The {State} of {Serverless} {Applications}: {Collection}, {Characterization}, and {Community} {Consensus}},
	volume = {48},
	copyright = {https://ieeexplore.ieee.org/Xplorehelp/downloads/license-information/IEEE.html},
	shorttitle = {The {State} of {Serverless} {Applications}},
	number = {10},
	journal = {IEEE Transactions on Software Engineering},
	author = {Eismann, Simon and Scheuner, Joel and Eyk, Erwin Van and Schwinger, Maximilian and Grohmann, Johannes and Herbst, Nikolas and Abad, Cristina L. and Iosup, Alexandru},
	month = oct,
	year = {2022},
	pages = {4152--4166},
}

@inproceedings{sahraei_xfaas_2023,
	address = {Koblenz Germany},
	title = {{XFaaS}: {Hyperscale} and {Low} {Cost} {Serverless} {Functions} at {Meta}},
	shorttitle = {{XFaaS}},
	booktitle = {Proceedings of the 29th {Symposium} on {Operating} {Systems} {Principles}},
	author = {Sahraei, Alireza and Demetriou, Soteris and Sobhgol, Amirali and Zhang, Haoran and Nagaraja, Abhigna and Pathak, Neeraj and Joshi, Girish and Souza, Carla and Huang, Bo and Cook, Wyatt and Golovei, Andrii and Venkat, Pradeep and Mcfague, Andrew and Skarlatos, Dimitrios and Patel, Vipul and Thind, Ravinder and Gonzalez, Ernesto and Jin, Yun and Tang, Chunqiang},
	month = oct,
	year = {2023},
	pages = {231--246},
}

@article{zahed2025efficient,
  title={An efficient function placement approach in serverless edge computing},
  author={Zahed, Atiya and Ghobaei-Arani, Mostafa and Esmaeili, Leila},
  journal={Computing},
  volume={107},
  number={3},
  pages={80},
  year={2025},
  publisher={Springer}
}

@article{rajan_review_2020,
	title = {A review on serverless architectures - function as a service ({FaaS}) in cloud computing},
	volume = {18},
	number = {1},
	journal = {TELKOMNIKA (Telecommunication Computing Electronics and Control)},
	author = {Rajan, Arokia Paul},
	month = feb,
	year = {2020},
	pages = {530},
}

@article{szalay_real-time_2023,
	title = {Real-{Time} {FaaS}: {Towards} a {Latency} {Bounded} {Serverless} {Cloud}},
	volume = {11},
	copyright = {https://ieeexplore.ieee.org/Xplorehelp/downloads/license-information/IEEE.html},
	shorttitle = {Real-{Time} {FaaS}},
	number = {2},
	journal = {IEEE Transactions on Cloud Computing},
	author = {Szalay, Márk and Mátray, Péter and Toka, László},
	month = apr,
	year = {2023},
	pages = {1636--1650},
}

@incollection{yilmaz_multivocal_2020,
	title = {A {Multivocal} {Literature} {Review} of {Function}-as-a-{Service} ({FaaS}) {Infrastructures} and {Implications} for {Software} {Developers}},
	booktitle = {Systems, {Software} and {Services} {Process} {Improvement}},
	author = {Yilmaz, Murat and Niemann, Jörg and Clarke, Paul and Messnarz, Richard},
	year = {2020},
	pages = {58--75},
}

@inproceedings{barcelona-pons_faas_2019,
	address = {Davis CA USA},
	title = {On the {FaaS} {Track}: {Building} {Stateful} {Distributed} {Applications} with {Serverless} {Architectures}},
	shorttitle = {On the {FaaS} {Track}},
	booktitle = {Proceedings of the 20th {International} {Middleware} {Conference}},
	
	author = {Barcelona-Pons, Daniel and Sánchez-Artigas, Marc and París, Gerard and Sutra, Pierre and García-López, Pedro},
	month = dec,
	year = {2019},
	pages = {41--54},
}

@inproceedings{pfandzelter_tinyfaas_2020,
	address = {Sydney, Australia},
	title = {{tinyFaaS}: {A} {Lightweight} {FaaS} {Platform} for {Edge} {Environments}},
	copyright = {https://ieeexplore.ieee.org/Xplorehelp/downloads/license-information/IEEE.html},
	shorttitle = {{tinyFaaS}},
	booktitle = {2020 {IEEE} {International} {Conference} on {Fog} {Computing} ({ICFC})},
	
	author = {Pfandzelter, Tobias and Bermbach, David},
	month = apr,
	year = {2020},
	pages = {17--24},
}

@book{tawarmalani_convexification_2002,
	address = {Boston, [Massachusetts]},
	series = {Nonconvex {Optimization} and its {Applications}},
	title = {Convexification and global optimization in continuous and mixed-integer nonlinear programming: theory, algorithms, software, and applications},
	shorttitle = {Convexification and global optimization in continuous and mixed-integer nonlinear programming},
	language = {eng},
	number = {Volume 65},
	author = {Tawarmalani, Mohit and Sahinidis, Nikolaos V.},
	year = {2002},
}

@article{schultz_twostage_1996,
	title = {Two‐stage stochastic integer programming: a survey},
	volume = {50},
	shorttitle = {Two‐stage stochastic integer programming},
	number = {3},
	journal = {Statistica Neerlandica},
	author = {Schultz, R. and Stougie, L. and Van Der Vlerk, M. H.},
	month = nov,
	year = {1996},
	pages = {404--416},
}

@article{marler_weighted_2010,
	title = {The weighted sum method for multi-objective optimization: new insights},
	volume = {41},
	copyright = {http://www.springer.com/tdm},
	shorttitle = {The weighted sum method for multi-objective optimization},
	number = {6},
	journal = {Structural and Multidisciplinary Optimization},
	author = {Marler, R. Timothy and Arora, Jasbir S.},
	month = jun,
	year = {2010},
	pages = {853--862},
}

@inproceedings{ngatchou_pareto_2005,
	address = {Arlington, Virginia, USA},
	title = {Pareto {Multi} {Objective} {Optimization}},
	booktitle = {Proceedings of the 13th {International} {Conference} on, {Intelligent} {Systems} {Application} to {Power} {Systems}},
	author = {Ngatchou, P. and Zarei, A. and El-Sharkawi, A.},
	year = {2005},
	pages = {84--91},
}

\end{document}